\documentclass[11pt,english]{article}
\usepackage{color}
\usepackage[T1]{fontenc}
\usepackage[latin9]{inputenc}
\usepackage[a4paper]{geometry}
\usepackage{textcomp}
\usepackage{amsthm}
\usepackage{amsmath}
\usepackage{amssymb}
\usepackage{esint}
\usepackage{bbm}
\usepackage{hyperref}
\usepackage{babel}
\usepackage{xcolor}

\usepackage{amscd}

\usepackage{amsfonts}
\usepackage{amsthm}
\usepackage{mathrsfs}
\usepackage{dsfont}

\usepackage{mathtools}

\usepackage{enumerate}
\usepackage{bm}
\usepackage{bbm}
\usepackage{verbatim}

\usepackage{enumitem}
\newlist{abbrv}{itemize}{1}
\setlist[abbrv,1]{label=,labelwidth=1.2in,align=parleft,itemsep=0.1\baselineskip,leftmargin=!}

\makeatletter

\DeclareFontEncoding{LGR}{}{}
\DeclareTextSymbol{\~}{LGR}{126}

\newcommand{\abs}[1]{\left| #1 \right|}

\newcommand{ \Hone }{H^1( \mathbb{R}^3)}

\newcommand{\scp}[2]{\big\langle #1 , #2 \big\rangle}

\newcommand{\norm}[1]{\left\| #1 \right\| }

\renewcommand{\Re}{\mathrm{Re}}
\renewcommand{\Im}{\mathrm{Im}}

\newcommand{\be}{\begin{equation}}
\newcommand{\ee}{\end{equation}}

\newtheorem{theorem}{Theorem}[section]
\newtheorem{lemma}{Lemma}[section]

\newtheorem{assumption}{Assumption}[section]

\theoremstyle{definition}

\theoremstyle{remark}
\newtheorem{remark}{Remark}[section]

\usepackage[colorinlistoftodos]{todonotes}

\allowdisplaybreaks[1]

\begin{document}


\title{The Landau-Pekar equations: Adiabatic theorem and accuracy }

\author{Nikolai Leopold, Simone Rademacher, Benjamin Schlein and Robert Seiringer
}

\maketitle

\begin{abstract}

\noindent

We prove an adiabatic theorem for the Landau-Pekar equations. This allows us to derive new results on the accuracy of their use as effective equations for the time evolution generated by the Fr\"ohlich Hamiltonian with large coupling constant $\alpha$.
In particular, we show that the time evolution of Pekar product states with coherent phonon field and the electron being trapped by the phonons is well approximated by the Landau-Pekar equations until times short compared to $\alpha^2$.

\end{abstract}

\noindent

\section{Introduction}

We are interested in the dynamics of an electron in a ionic crystal. For situations in which the extension of the electron is much larger than the lattice spacing, Fr\"ohlich \cite{froehlich} derived a model which treats the crystal as a continuous medium and describes the polarization of the lattice as the excitations of a quantum field, called phonons.
If the coupling between the electron and the phonons
is large it is expected that the dynamics of the system can be approximated by the Landau-Pekar equations, a set of nonlinear differential equations which model the phonons by means of a classical field. The coupling parameter of the Fr\"ohlich model enters into the Landau-Pekar equations leading to a separation of time scales of the electron and the phonon field. This phenomenon, often referred to as adiabatic decoupling \cite{teufel}, is believed to be responsible for the classical behavior of the radiation field. The physical picture one has in mind is that the electron is trapped in a cloud of slower phonons which increase the effective mass of the electron \cite{liebseiringer}.

The goal of this paper is to compare the time evolution generated by the Fr\"ohlich Hamiltonian with the Landau-Pekar equations and to give a quantitative justification of the applied approximation. In particular, we will consider the evolution of factorized initial data, with a coherent phonon field and an electron trapped by the phonons and minimizing the corresponding energy. For such initial data, we show that the Landau-Pekar equations provide a good approximation of the dynamics, up to times short compared to $\alpha^2$, with $\alpha$ denoting the coupling between the electron and phonons (space and time units are chosen so that the electron is initially localized in a volume of order one and has speed of order one). This result improves previous bounds in \cite{frankschlein, frankgang}, which only hold up to times of order $\alpha$ (but for more general initial data). Also, it extends the findings of \cite{griesemer}, which show a result similar to ours but only for initial data minimizing the Pekar energy functional (in this case, the solution of the Landau-Pekar equations remains constant). To prove our bound, we establish an adiabatic theorem for the solution of the Landau-Pekar equations. The idea of considering states with the electron trapped by the phonon field and showing an adiabatic theorem was first proposed in \cite{frankgang2,frankgang3}, where an adiabatic theorem is proved for a one-dimensional version of the Landau-Pekar equations. Apart from the restriction to the one-dimensional setting, the adiabatic theorem in \cite{frankgang2,frankgang3} differs from ours, because it compares the full solution of the Landau-Pekar equations with the solution of a limiting system of equations, independent of $\alpha$ (in Theorem \ref{thm:adiabatic} below, on the other hand, we only compare the electron wave function with the ground state of the Schr\"odinger operator associated with the phonon field determined by the Landau-Pekar equations; this is sufficient for our purposes).

\section{Model and Results}

We consider the  Fr\"ohlich model which describes the interaction between an electron and a quantized phonon field. The state of the phonon field is represented by an element of the bosonic Fock space $\mathcal{F} \coloneqq \bigoplus_{n \geq 0} L^2(\mathbb{R}^3)^{\otimes_s^n}$, where the subscript $s$ indicates symmetry under the interchange of variables.
The system is described by elements $\Psi_t \in \mathcal{H}$ of the Hilbert space 
\begin{align}
\mathcal{H} \coloneqq L^2(\mathbb{R}^3) \otimes \mathcal{F}.
\end{align}
Its time evolution is governed by the Schr\"odinger equation
\begin{align}
\label{eq: Schroedinger equation}
i \partial_t \Psi_t = H_{\alpha} \Psi_t
\end{align}
with the Fr\"ohlich Hamiltonian
\begin{align}
\label{eq: Froehlich Hamiltonian}
H_{\alpha} \coloneqq  - \Delta + \int d^3k \, \abs{k}^{-1}  \big( e^{ik \cdot x} a_k + e^{-ik \cdot x}  a_k^*  \big) + \int d^3k \, a_k^* a_k .
\end{align}
Here,  $a_k^*$ and $a_k$ are the creation and annihilation operators in the Fock space $\mathcal{F}$, satisfying the commutation relations
\begin{align}
[a_k, a_{k'}^*] &= \alpha^{-2} \delta(k - k') ,
\quad 
[a_k, a_{k'}] = [a_k^*, a_{k'}^*] = 0
\quad \text{for all} \, k, k' \in \mathbb{R}^3,
\end{align}
for a coupling constant $\alpha >0$. One should note that the Hamiltonian is written in the strong coupling units, which gives rise to the $\alpha$ dependence in the commutation relations. These units are related to the usual ones by rescaling all lengths by $\alpha$ and time by $\alpha^2$, see \cite[Appendix A]{frankschlein}.
We will be interested in the limit $\alpha \rightarrow \infty$. 
Motivated by Pekar's Ansatz, we consider the evolution of initial states of product form
\begin{align}
\label{eq:product}
\psi_0 \otimes W(\alpha^2 \varphi_0) \Omega.
\end{align}
Here $\Omega$ denotes the vacuum of the Fock space $\mathcal{F}$ and $W(f)$ for $f \in L^2 ( \mathbb{R}^3)$ denotes the Weyl operator given by
\begin{align}
\label{eq: definition Weyl operator}
W(f) = \exp \left[ \int d^3k \, \left( f(k) a_k^* - \overline{f(k)} a_k  \right)  \right]  .
\end{align}
Note that the Weyl operator is unitary and satisfies the shifting property with respect to the creation and annihilation operator, i.e.
\begin{align}
\label{eq:Weyl_comm}
W^*(f)a_k W(f) = a_k + \alpha^{-2} f(k), \quad W^*(f)a^*_k W(f) = a^*_k + \alpha^{-2} \overline{f(k)}
\end{align}
for all $f \in L^2( \mathbb{R}^3)$.
Due to the interaction the system will develop correlations between the electron and the radiation field and the solution of \eqref{eq: Schroedinger equation} will no longer be of product form. However, for an appropriate class of initial states we will show that it can be approximated up to times short compared to $\alpha^2$ (in the limit of large $\alpha$) by a product state $\psi_t \otimes W(\alpha^2 \varphi_t) \Omega$, with $(\psi_t,\varphi_t)$ being a solution of the Landau-Pekar equations
\begin{align}
\label{eq: Landau Pekar equations}
\begin{cases}
i \partial_t \psi_t &= \left[ - \Delta + \int d^3k \, \abs{k}^{-1} \left( e^{ ik \cdot x} \varphi_t(k) + e^{- ik \cdot x} \overline{\varphi_t(k)}  \right) \right] \psi_t(x) , \\
i \alpha^2 \partial_t \varphi_t(k) &= \varphi_t(k) + \abs{k}^{-1} \int d^3x \, e^{- ik \cdot x} \abs{\psi_t(x)}^2 
\end{cases}
\end{align}
with initial data $(\psi_0,\varphi_0)$. We define for $\varphi \in L^2(\mathbb{R}^3)$  the potential 
\begin{align}
\label{eq: definition potential}
 \; V_{\varphi} (x) = \int d^3k \, |k|^{-1}  \left( \varphi (k) e^{ik \cdot x} + \overline{\varphi (k)} e^{-ik \cdot x} \right) .
\end{align}
We are interested, in particular, in initial data of the form \eqref{eq:product} where the phonon field $\varphi$ is such that the Schr\"odinger operator
\begin{align}
\label{eq: hphit}
h_{\varphi} &\coloneqq - \Delta + V_{\varphi}
\end{align}
has a non-degenerate eigenvalue at the bottom of its spectrum separated from the rest of the spectrum by a gap, and the electron wave function $\psi$ is a ground state vector of \eqref{eq: hphit}.


\begin{assumption}
\label{assumptions}
Let $\varphi_0 \in L^2(\mathbb{R}^3)$ such that
\begin{align}
e(\varphi_0) \coloneqq  \inf \lbrace\langle \psi, h_{\varphi_0} \psi \rangle : \psi \in H^1( \mathbb{R}^3), \| \psi \|_2 =1 \rbrace <0 .
\end{align}
\end{assumption}
This assumption ensures the existence of a unique positive ground state vector $\psi_{\varphi_0}$ for $h_{\varphi_0}$ with corresponding eigenvalue separated from the rest of the spectrum by a gap of size $\Lambda (0) >0$.  If we then consider solutions of \eqref{eq: Landau Pekar equations} with initial data $(\psi_{\varphi_0},\varphi_0)$  the spectral gap can be shown to persist at least for times of order $\alpha^2$. 
 \begin{lemma}
\label{lemma: spectral gap}
Let $\varphi_0$ satisfy Assumption \ref{assumptions} and let $( \psi_t, \varphi_t) \in H^1( \mathbb{R}^3) \times L^2( \mathbb{R}^3)$ denote the solution of the Landau-Pekar equations with initial value $( \psi_{\varphi_0}, \varphi_0 )$. Moreover, let 
\begin{align}
\Lambda(t) := \inf_{\substack{\lambda \in \mathrm{spec}( h_{\varphi_t})\\ \lambda \not= e(\varphi_t)}} \vert e(\varphi_t) - \lambda \vert.
\end{align}
Then, for all $\Lambda$ with $0 < \Lambda < \Lambda(0) $ there is a constant $C_{\Lambda} >0$ such that, for all $|t| \leq C_\Lambda \alpha^2$, the Hamiltonian $h_{\varphi_t}$ has a  unique positive and normalized ground state $\psi_{\varphi_t}$ with eigenvalue $e( \varphi_t)<0$, which is separated from the rest of the spectrum by a gap  of size $\Lambda (t) \geq \Lambda$.
\end{lemma}
The Lemma is proven in Subsection \ref{sec: minimizer}.  Using the persistence of the spectral gap, we can prove the following adiabatic theorem for the solution of the Landau-Pekar equations \eqref{eq: Landau Pekar equations}. As mentioned in the introduction, the idea of such result is based on \cite{frankgang2,frankgang3}, where an adiabatic theorem for the Landau-Pekar equations in one dimension is proved. 


\begin{theorem}
\label{thm:adiabatic}
Let $T >0$, $\Lambda >0$ and $( \psi_t, \varphi_t) \in H^1( \mathbb{R}^3) \times L^2( \mathbb{R}^3)$ denote the solution of the Landau-Pekar equations with initial value $( \psi_{\varphi_0}, \varphi_0 ) \in H^1( \mathbb{R}^3) \times L^2( \mathbb{R}^3)$. Assume that the Hamiltonian $h_{\varphi_t}$ has a unique positive and normalized ground state $\psi_{\varphi_t}$ and a spectral gap of size $\Lambda(t) > \Lambda$ for all $|t| \leq T$.
Then
\begin{align}
\label{eq: adiabatic theorem}
\| \psi_t - e^{- i \int_0^t du \; e( \varphi_u)} \psi_{\varphi_t} \|_2^2 \leq C \Lambda^{-4} \alpha^{-4} \left( 1 +  \left( 1 + \Lambda^{-1} \right)^2 \alpha^{-4} |t|^2 \right), \hspace{0.3cm} \forall |t| \leq  T.
\end{align}
\end{theorem}

\begin{remark}
\label{remark: adiabatic theorem for short times}
One also has
$
 \| \psi_t - e^{- i \int_0^t du \; e( \varphi_u)} \psi_{\varphi_t} \|_{2}^2 \leq C  \alpha^{-2} \Lambda^{-1} \abs{t} 
$ for all $\abs{t} \leq T$.
\end{remark}

\begin{remark}
Note that the proof of the theorem only requires the existence of the spectral gap $\Lambda>0$.  Assuming $\Lambda$ to be of order one for times of order $\alpha^4$, the theorem shows that  $\psi_t$ is well approximated by the ground state $\psi_{\varphi_t}$ for any $|t| \ll \alpha^{4}$.
\end{remark}
\begin{remark}
Lemma \ref{lemma: spectral gap} shows that  the existence of the ground state and the spectral gap for all times $\abs{t} \leq C_{\Lambda} \alpha^2$   can be inferred from Assumption \ref{assumptions}.
 In this case, \eqref{eq: adiabatic theorem} is valid for all $\abs{t} \leq C_{\Lambda} \alpha^2$ without any assumptions on $h_{\varphi_t}$ and $\Lambda(t)$ at times $t >0$. Theorem \ref{thm:adiabatic} implies that there exists $C_\Lambda, \widetilde{C}_\Lambda >0$ such that
 \begin{align}
\label{eq: adiabatic theorem improved bound}
 \| \psi_t - e^{- i \int_0^t du \; e( \varphi_u)} \psi_{\varphi_t} \|_{2}^2 &\leq \widetilde{C_\Lambda} \alpha^{-4}
 \end{align}
for all $|t| \leq C_\Lambda \alpha^2$. 
 \end{remark} 


Using Theorem \ref{thm:adiabatic} we can show that the Landau-Pekar equations (\ref{eq: Landau Pekar equations}) provide a good approximation to the solution of the Schr\"odinger equation (\ref{eq: Schroedinger equation}), for initial data of the form (\ref{eq:product}), with $\varphi_0$ satisfying Assumption \ref{assumptions} and with $\psi_0 = \psi_{\varphi_0}$ being the ground state of the operator $h_{\varphi_0}$ defined as in (\ref{eq: hphit}).

\begin{theorem}
\label{theorem: main theorem}
Let $\varphi_0$ satisfy Assumption \ref{assumptions} and $\alpha_0 >0$.   Let $( \psi_t, \varphi_t) \in H^1( \mathbb{R}^3) \times L^2( \mathbb{R}^3)$ denote the solution of the Landau-Pekar equations with initial data $( \psi_{\varphi_0}, \varphi_0 ) \in H^1( \mathbb{R}^3) \times L^2( \mathbb{R}^3)$ and
\begin{align}
\omega(t) \coloneqq \alpha^2 \Im \scp{\varphi_t}{\partial_t \varphi_t} + \norm{\varphi_t}_2^2 . 
\end{align}
Then, there exists a constant $C>0$  such that 
\begin{align}
\label{eq: bound main theorem}
\norm{e^{-i  H_{\alpha}  t} \psi_{\varphi_0} \otimes W(\alpha^2 \varphi_0) \Omega - e^{ - i \int_0^t du \, \omega(u)} \psi_t \otimes W(\alpha^2 \varphi_t) \Omega}
&\leq C \alpha^{-1}  \abs{t}^{1/2}
\end{align}
for all $\alpha \geq \alpha_0$. The constant $C>0$ depends only on $\alpha_0>0$ and the initial condition, i.e. $( \psi_{\varphi_0}, \varphi_0 ) \in H^1( \mathbb{R}^3) \times L^2( \mathbb{R}^3)$ and the spectral gap $\Lambda (0)$ of $h_{\varphi_0}$.
\end{theorem}

\begin{remark}
Theorem \ref{theorem: main theorem} shows that the Pekar ansatz is a good approximation for times small compared to $\alpha^2$. Note that even though $(\psi_t, \varphi_t)$ stay close to their initial values for these times (as shown in Theorem \ref{thm:adiabatic}), it is essential to use the time-evolved version in \eqref{eq: bound main theorem}. This is due to the large factor $\alpha^2$ in the Weyl operator $W(\alpha^2 \varphi_t)$, which leads to a very sensitive behavior of the state on $\varphi_t$.

During the publication process of this article it has been shown \cite[Theorem I.5]{LMRSS2021} that the Landau--Pekar equations approximately describe the Fr\"ohlich evolution of the electron- and one-phonon reduced density matrices up to times of order $\alpha^2$.  In order to approximate the many-body state in norm at times of order $\alpha^2$ it is, however, necessary to take quantum fluctuations into account; see \cite[Theorem I.3 and Remark I.7]{LMRSS2021} and \cite{M2021} for initial data minimizing the Pekar energy functional. In this context we would also like to mention new results about the persistence of the spectral gap for the Landau--Pekar equations \cite{FRS2021}.
\end{remark}

A first rigorous result concerning the effective evolution of the Fr\"ohlich polaron was obtained in \cite{frankschlein}, where the product $\psi_t \otimes W(\alpha^2 \varphi_0) \Omega$, with $\psi_t$ solving the linear equation
\begin{align}
 i\partial_t \psi_t = -\Delta \psi_t + \int dk \, |k|^{-1} \left( e^{ik \cdot x} \varphi_0 (k) + e^{-ik \cdot x} \overline{\varphi_0 (k)} \right) \psi_t (x) 
\end{align}
was proven to give a good approximation for the solution $e^{-i H_\alpha t} \psi_0 \otimes W (\alpha^2 \varphi_0) \Omega$ of the Schr\"odinger equation (\ref{eq: Schroedinger equation}), up to times of order one\footnote{In fact, a simple modification of the Gronwall argument in \cite{frankschlein} leads to convergence for times $|t| \ll \alpha$.}. This result was improved in \cite{frankgang}, where convergence towards (\ref{eq: Landau Pekar equations}) was established for all times $|t| \ll \alpha$ (the analysis in that work also gives more detailed information on the solution of Schr\"odinger equation (\ref{eq: Schroedinger equation}); in particular, it implies convergence of reduced density matrices). Note that the results of \cite{frankschlein,frankgang} hold for general initial data of the form (\ref{eq:product}), with no assumption on the relation between the initial electron wave function $\psi_0$ and the initial phonon field $\varphi_0$. Theorem \ref{theorem: main theorem} shows therefore that, under the additional assumption that $\psi_0 = \psi_{\varphi_0}$ the Landau-Pekar equations (\ref{eq: Landau Pekar equations}) provide a good approximation to (\ref{eq: Schroedinger equation}) for longer times (times short compared to $\alpha^2$). In this sense, Theorem \ref{theorem: main theorem} also extends the result of \cite{griesemer}, where the validity of the Landau-Pekar equations was established for times short compared to $\alpha^2$, but only for initial data $(\psi_0, \varphi_0)$ minimizing the Pekar energy functional (in this case, the solution of (\ref{eq: Landau Pekar equations}) is stationary, i.e. $(\psi_t, \varphi_t) = (e^{- i t \, e (\varphi_0)} \psi_0, \varphi_0)$ for all $t$). In fact, similarly to the analysis in \cite{griesemer}, we use the observation that the spectral gap above the ground state energy of $h_{\varphi_t}$ allows us to obtain bounds that are valid on longer time scales (it allows us to integrate by parts, after (\ref{eq:adia_1}) and after (\ref{eq:parts2}); this step is crucial to save a factor of $t$). 
The classical behavior of a quantum field does not only appear in the strong coupling limit of the Fr\"ohlich polaron but has also been studied in other situations.
In \cite{ginibrenironivelo} it was shown in case of the Nelson model that a quantum scalar field behaves classically in a certain limit where the number of field bosons becomes infinite while the coupling constant tends to zero.
The emergence of classical radiation was also proven for the Nelson model with ultraviolet cutoff \cite{falconi, leopoldpetrat}, the renormalized Nelson model \cite{ammarifalconi} and the Pauli-Fierz Hamiltonian \cite{leopoldpickl} in situations in which a large number of particles weakly couple to the radiation field. The articles \cite{davies_1979, hiroshima_1998, teufel2} revealed in addition that quantum fields can sometimes be replaced by two-particle interactions if the particles are much slower than the bosons of the quantum field.

\section{Preliminaries}

In this section, we collect  properties of the Landau-Pekar equations that are used in the proofs of Theorem \ref{thm:adiabatic} and Theorem \ref{theorem: main theorem}.
For $\psi \in L^2( \mathbb{R}^3)$ we define the function
\begin{align}
\label{eq: sigma}
\sigma_{\psi}(k) = |k|^{-1} \int d^3x \, e^{- i k \cdot x} |\psi (x)|^2  . 
\end{align}
The first and second derivative of the potential $V_{\varphi_t}$ are given by\footnote{We use the notation $\dot{f}$ to denote the derivative of a function $f$ with respect to time.} 
\begin{align}
\label{eq: time derivative of the potential}
\partial_t V_{\varphi_t}(x) 
= V_{\dot{\varphi}_t}(x)
&= - \alpha^{-2} \int d^3k \, \abs{k}^{-1} \big( e^{i k \cdot x}  i \varphi_t(k)
+ e^{- ik \cdot x} \overline{i \varphi_t(k)} \big)  
\nonumber \\
&\quad - i \alpha^{-2} \int d^3k \, \abs{k}^{-1} 
\big(  e^{i k \cdot x} \sigma_{\psi_t}(k) - e^{- i k \cdot x} \sigma_{\psi_t}(-k)  \big)
= - \alpha^{-2} V_{i \varphi_t}(x)
\end{align}
and 
\begin{align}
\label{eq: derivative potential with i-varphi}
\partial_t V_{i \varphi_t}(x) 
&= V_{i \dot{\varphi}_t}(x)
=  V_{\alpha^{-2} ( \varphi_t + \sigma_{\psi_t} )}(x)   
= \alpha^{-2} V_{\varphi_t}(x) + \alpha^{-2} V_{\sigma_{\psi_t}}(x) .
\end{align}
We define the energy functional $\mathcal{E}: H^1(\mathbb{R}^3) \times L^2(\mathbb{R}^3) \rightarrow \mathbb{R}$
\begin{align}
\mathcal{E} ( \psi, \varphi ) = \langle \psi, h_{\varphi} \psi \rangle + \| \varphi \|_2^2. 
\end{align}
Using standard methods one can show that the Landau-Pekar equations are well posed and that the energy $\mathcal{E} ( \psi_t , \varphi_t )$ is conserved if $(\psi_t, \varphi_t)$ is a solution of \eqref{eq: Landau Pekar equations}. 
For a proof of the following Lemma see \cite[Appendix C]{frankgang}.
\begin{lemma}[\cite{frankgang}, Lemma 2.1]
\label{lemma: well posedness LP}
For any $ ( \psi_0, \varphi_0) \in H^1( \mathbb{R}^3) \times L^2( \mathbb{R}^3)$, there is a unique global solution $( \psi_t, \varphi_t)$ of the Landau-Pekar equations \eqref{eq: Landau Pekar equations}. The following conservation laws hold true
\begin{align}
\| \psi_t \|_{2} = \| \psi_0 \|_{2} \hspace{0.3cm} and \hspace{0.3cm} \mathcal{E} ( \psi_t, \varphi_t ) = \mathcal{E} ( \psi_0, \varphi_0 ) \hspace{0.3cm} \forall t \in \mathbb{R}^3.
\end{align}
Moreover, there exists a constant $C$ such that 
\begin{align}
\| \psi_t \|_{H^1( \mathbb{R}^3)} \leq C, \hspace{0.3cm} \| \varphi_t \|_{2} \leq C 
\end{align}
for all $\alpha >0$ and all $t \in \mathbb{R}$.
\end{lemma}
%


The next Lemma (also proven in  \cite[Appendices B and C]{frankgang}) collects some properties of quantities occurring in the Landau-Pekar equations. 

\begin{lemma}
\label{lemma: Potental} For  $V_\varphi$ being defined as in \eqref{eq: definition potential} 
there exists a constant $C>0$ such that for every $ \psi \in H^1( \mathbb{R}^3)$ and $ \varphi \in L^2( \mathbb{R}^3)$ 
\begin{align}
\| V_\varphi \|_6 \leq C  \| \varphi \|_2
 \hspace{0.3cm} 
 and \hspace{0.3cm} 
\| V_\varphi \psi \|_2 \leq  C  \| \varphi \|_2 \; \|\psi \|_{H^1( \mathbb{R}^3)} .
\end{align}
Furthermore, for every $\delta >0$ there exists $C_\delta >0$ such that 
\begin{align}
\pm V_\varphi  \leq - \delta \Delta + C_\delta, 
\end{align}
thus there exists $C>0$ such that
\begin{align}
\label{eq: bound for h-varphi}
- \frac{1}{2} \Delta - C \leq h_{\varphi} \leq - 2 \Delta + C.
\end{align}
Let $\sigma_\psi$ be defined as in \eqref{eq: sigma}. Then, there exists $C>0$ such that
\begin{align}
\label{eq: bound sigma-psi}
\| \sigma_\psi \|_{2} \leq C \| \psi \|_{H^1( \mathbb{R}^3)}^2 .
\end{align}
\end{lemma}

\begin{remark}
Let $T >0$, $\Lambda >0$ and $( \psi_t, \varphi_t) \in H^1( \mathbb{R}^3) \times L^2( \mathbb{R}^3)$ denote the solution of the Landau-Pekar equations with initial value $( \psi_{\varphi_0}, \varphi_0 ) \in H^1( \mathbb{R}^3) \times L^2( \mathbb{R}^3)$. Assume that the Hamiltonian $h_{\varphi_t}$ has a unique positive and normalized ground state $\psi_{\varphi_t}$ and a spectral gap of size $\Lambda(t) > \Lambda$ for all $t \leq T$.  Lemma \ref{lemma: Potental} then implies the existence of constant such that
\begin{align}
\label{eq: h1 minimizer}
\| \psi_{\varphi_t} \|_{H^1( \mathbb{R}^3)} \leq C    
\quad \forall \abs{t} \leq T. 
\end{align}
\end{remark}

\begin{proof}[Proof of Lemma \ref{lemma: Potental}]
Recall the definition \eqref{eq: definition potential} of the potential $V_\varphi$. We write 
\begin{align}
V_{\varphi} (x) =   2^{3/2} \pi^{-1/2} \; \Re \int \frac{d^3y}{|x-y|^2} \check{\varphi} (y) ,
\end{align}
where $\check{\varphi}$ denotes the inverse Fourier transform defined for $\varphi \in L^1( \mathbb{R}^3)$ through
\begin{align}
 \check{\varphi}(x)  = (2 \pi)^{-3/2} \int d^3k \, e^{ ik \cdot x} \varphi(k).
\end{align}
The first inequality follows directly from the Hardy-Littlewood-Sobolev inequality
\begin{align}
\| V_\varphi \|_6 \leq C \| \varphi \|_2.
\end{align}
In order to prove the second inequality we use the first one and the H\"older inequality
\begin{align}
\| V_\varphi \psi \|_2 \leq   \| V_\varphi\|_6 \; \| \psi \|_3 \leq  C \| \varphi \|_2 \; \| \psi \|_3 .
\end{align}
Since the interpolation inequality together with the Sobolev inequality implies
\begin{align}
\| \psi \|_3 \leq \| \psi \|_2^{1/2} \; \| \psi \|_6^{1/2} \leq \|\psi \|_2^{1/2} \; \|\nabla \psi \|_2^{1/2},
\end{align}
we obtain
\begin{align}
\| V_\varphi \psi \|_2 \leq  C  \| \varphi \|_2 \; \| \nabla \psi \|_2^{1/2} \; \| \psi \|_2^{1/2} \leq
C
 \| \varphi \|_2 \| \psi \|_{H^1( \mathbb{R}^3)} .
\end{align}
The second operator inequality follows again from the Sobolev inequality. For this let $\psi \in H^1( \mathbb{R}^3)$, then for $\varepsilon >0$
\begin{align}
\langle \psi, V_\varphi \psi \rangle \leq  C \| V_\varphi \|_6 \| \psi \|_{12/5}^2 \leq C \| V_\varphi \|_6 \left( \varepsilon \| \nabla \psi \|_2^2 + \varepsilon^{-1} \| \psi \|_2^2 \right),
\end{align}
where we used the interpolation inequality. The first inequality of the Lemma implies 
\begin{align}
\pm \langle \psi, V_\varphi \psi \rangle \leq \langle \psi, \left( - \varepsilon \Delta + C_\varepsilon \right) \psi \rangle,
\end{align}
and \eqref{eq: bound for h-varphi} follows.

The last inequality of the Lemma follows from the observation that 
\begin{align}
\| \sigma_\psi \|^2_{2} &= \int \frac{dk}{|k|^2} \int dx dy \;| \psi (x)|^2 | \psi (y)|^2 e^{ik \cdot (x-y)} 
\nonumber \\
&= 2 \pi^2   \int dxdy \frac{| \psi (x)|^2 | \psi(y)|^2}{|x-y|}  \leq C \|  |\psi |^2  \|_{6/5}^2 ,
\end{align}
where we used again the Hardy-Littlewood-Sobolev inequality. As before, the interpolation and the Sobolev inequality imply \eqref{eq: bound sigma-psi}.

\section{Proof of the adiabatic theorem}

\subsection{The ground state $\psi_{\varphi_t}$}

\label{sec: minimizer}
 
Before proving the adiabatic theorem, we show that the spectral gap of $h_{\varphi_t}$ does not close for times of order $\alpha^2$. In particular, we show the existence of the ground state $\psi_{\varphi_t}$.

\begin{lemma}
\label{lemma: existence minimizer}
Let $\varphi_0$ satisfy Assumption \ref{assumptions}. Then, there exists a unique positive and normalized ground state $\psi_{\varphi_0}$ of $h_{\varphi_0} = - \Delta + V_{\varphi_0}$. 
Let $( \psi_t, \varphi_t) \in H^1( \mathbb{R}^3) \times L^2( \mathbb{R}^3)$ denote the solution of the Landau-Pekar equations \eqref{eq: Landau Pekar equations} with initial value $( \psi_{\varphi_0}, \varphi_0 )$. There exists $C>0$ such that for all $|t|  \leq C \alpha^2$, there exists a unique, positive and normalized ground state $\psi_{\varphi_t}$ of $h_{\varphi_t} = - \Delta + V_{\varphi_t}$ with corresponding eigenvalue $e( \varphi_t) < 0$. It satisfies 
\begin{align}
\partial_t \psi_{\varphi_t} = \alpha^{-2} R_t  V_{i \varphi_t} \psi_{\varphi_t}
 \hspace{0.3cm} \mathrm{with} \hspace{0.3cm} R_t = q_t (h_{\varphi_t} - e(\varphi_t))^{-1} q_t ,
\end{align}
where $q_{t} = 1- | \psi_{\varphi_t} \rangle \langle \psi_{\varphi_t} |$ denotes the projection onto the subspace of $L^2 ( \mathbb{R}^3)$ orthogonal to the span of $\psi_{\varphi_t}$. 
\end{lemma}

\begin{proof}
Lemma \ref{lemma: well posedness LP}  and Lemma \ref{lemma: Potental} imply that $V_{\varphi_t} \in L^6( \mathbb{R}^3)$ for all $t \in \mathbb{R}$. The existence of the ground state $\psi_{\varphi_t}$ at time $t=0$ then follows from the negativity of the infimum of the spectrum; see \cite[Theorem 11.5]{liebloss}. In order to prove the existence of the ground state $\psi_{\varphi_t}$ of $h_{\varphi_t}$ at later times, it suffices to show that $e( \varphi_t)$ is negative. For this we pick the ground state $\psi_{\varphi_0}$ at time $t=0$ and estimate
\begin{align}
\label{eq: bound lowest eigenvalue}
e( \varphi_t) &\leq \langle \psi_{\varphi_0}, ( - \Delta + V_{\varphi_t} ) \psi_{\varphi_0} \rangle = e( \varphi_0) - \alpha^{-2} \int_0^t ds \; \langle \psi_{\varphi_0}, V_{i \varphi_s} \psi_{\varphi_0} \rangle
\nonumber \\
 &\leq e( \varphi_0) + C \int_0^t ds\;   \alpha^{-2} \| \varphi_s \|_2 \; \| \psi_{\varphi_0} \|_{H^1( \mathbb{R}^3)}^2 \leq e( \varphi_0) + C |t| \alpha^{-2},
\end{align}
by means of \eqref{eq: time derivative of the potential}, Lemma \ref{lemma: well posedness LP}  and Lemma \ref{lemma: Potental}.
Thus if we restrict our consideration to times $|t| < C^{-1} |e( \varphi_0)|  \alpha^{2} $, we conclude that $e( \varphi_t) <0$. The ground state $\psi_{\varphi_t}$ satisfies
\begin{align}
\label{eq:deriv minimizer 1}
0= \left(h_{\varphi_t} - e( \varphi_t ) \right) \psi_{\varphi_t}
\end{align}
Differentiating both sides of the equality with respect to the time variable leads to
\begin{align}
\label{eq:deriv minimizer 2}
0 = \left( \dot{h}_{\varphi_t} - \dot{e}( \varphi_t) \right) \psi_{\varphi_t} +   \left(h_{\varphi_t} - e( \varphi_t ) \right) \dot{\psi}_{\varphi_t}.
\end{align}
On the one hand
$\dot{h}_{\varphi_t} = V_{\dot{\varphi}_t} =- \alpha^{-2} V_{i \varphi_t}$ by means of \eqref{eq: time derivative of the potential} and on the other hand, the Hellmann-Feynman theorem implies
\begin{align}
\label{eq:  time-derivative lowest eigenvalue}
\dot{e}( \varphi_t) = \langle \psi_{\varphi_t} , \dot{h}_{\varphi_t} \psi_{\varphi_t} \rangle =- \alpha^{-2} \langle \psi_{\varphi_t} ,V_{i \varphi_t} \psi_{\varphi_t} \rangle 
\end{align}
so that \eqref{eq:deriv minimizer 2} becomes
\begin{align}
0= -\alpha^{-2} q_t V_{i \varphi_t} \psi_{\varphi_t} + \left(h_{\varphi_t} - e( \varphi_t ) \right) \dot{\psi}_{\varphi_t}.
\end{align}
Since $\psi_{\varphi_t}$ is chosen  to be real and normalized for all $t \in \mathbb{R}$, it follows that $\scp{\psi_{\varphi_t}}{\dot{\psi}_{\varphi_t}} = 0$ for all $t \in \mathbb{R}$. Hence,  
\begin{align}
\dot{\psi}_{\varphi_t} =  \alpha^{-2} q_t ( h_{\varphi_t} - e( \varphi_t))^{-1} q_t V_{i \varphi_t} \psi_{\varphi_t}  = \alpha^{-2} R_t V_{i \varphi_t} \psi_{\varphi_t},
\end{align}
completing the proof.
\end{proof}


Using the Lemma above, we prove Lemma \ref{lemma: spectral gap}. 

\begin{proof}[Proof of Lemma \ref{lemma: spectral gap}.]
By the min-max principle \cite[Theorem 12.1]{liebloss}, the first excited eigenvalue of $h_{\varphi_t}$ (or the bottom of the essential spectrum) is given by
\begin{align}
\label{eq:e1}
e_1(t) = \inf_{\substack{A \subset L^2( \mathbb{R}^3) \\ \mathrm{dim} A=2}} \sup_{\substack{\psi \in A \\ \| \psi \|_2 =1} } \langle\psi, h_{\varphi_t} \psi \rangle .
\end{align}
For any $\psi \in L^2( \mathbb{R}^3)$ with $\| \psi \|_2 =1$ we have by Lemma \ref{lemma: Potental}
\begin{align}
\langle \psi, h_{\varphi_t} \psi \rangle =& \langle \psi, h_{\varphi_0} \psi \rangle - \alpha^{-2} \int_0^t ds \; \langle \psi, V_{i \varphi_s} \psi \rangle
\nonumber \\
\geq& \langle \psi, h_{\varphi_0} \psi \rangle - C |t| \alpha^{-2} \sup_{|s| \leq |t|} \| \varphi_s \|_2 \| \nabla \psi \|_2^2  - C |t| \alpha^{-2} \sup_{|s| \leq |t|} \| \varphi_s \|_2 \| \psi\|_2^2
\nonumber \\
\geq& (1-2 C |t| \alpha^{-2} ) \langle \psi, h_{\varphi_0} \psi \rangle - C |t| \alpha^{-2} 
\end{align}
Inserting in \eqref{eq:e1},  we conclude that 
\begin{align}
e_1(t) \geq (1-2C|t|\alpha^{-2}) e_1(0) - C |t| \alpha^{-2} .
\end{align}
With $e(\varphi_t) \leq e(\varphi_0) + C |t| \alpha^{-2}$ (see \eqref{eq: bound lowest eigenvalue}), we obtain
\begin{align}
\Lambda (t) \geq \Lambda (0) - C |t| \alpha^{-2},
\end{align}
completing the proof.
\end{proof}
 

Using the persistence of the spectral gap, the resolvent $R_t = q_t \left( h_{\varphi_t}- e( \varphi_t) \right)^{-1} q_t$ can be estimated as follows.

\begin{lemma}
\label{lemma: resolvent}
Let $T >0$, $\Lambda >0$ and $( \psi_t, \varphi_t) \in H^1( \mathbb{R}^3) \times L^2( \mathbb{R}^3)$ denote the solution of the Landau-Pekar equations with initial value $( \psi_{\varphi_0}, \varphi_0 ) \in H^1( \mathbb{R}^3) \times L^2( \mathbb{R}^3)$. Assume that the Hamiltonian $h_{\varphi_t}$ has a unique positive and normalized ground state $\psi_{\varphi_t}$ with $e(\varphi_t) < 0$ and a spectral gap of size $\Lambda(t) > \Lambda$ for all $t \leq T$.
Then, for all $|t| \leq T$
\begin{align}
\| R_t \| \leq \Lambda^{-1},  \hspace{0.5cm} \| ( - \Delta +1)^{1/2} R_t^{1/2} \| \leq C ( 1 + \Lambda^{-1})^{1/2}, 
\end{align}
and
\begin{align}
\| \dot{R}_t \| \leq C  \Lambda^{-3/2}\alpha^{-2}  ( 1 + \Lambda^{-1})^{1/2} 
\end{align}
with $C >0$ depending only on $\varphi_0$.
\end{lemma}

\begin{proof}
Since the spectral gap is at least of size $\Lambda>0$ for times $|t| \leq T$, it follows that
\begin{align}
\| R_t \| \leq \Lambda^{-1}. 
 \end{align}
To prove the second inequality, we estimate for arbitrary $\psi \in L^2( \mathbb{R}^3)$
\begin{align}
\|( - \Delta +1)^{1/2}  R_t^{1/2} \psi \|^2_2 &= \langle  \psi, R_t^{1/2} ( - \Delta +1) R_t^{1/2} \psi \rangle 
\leq C  \langle \psi, R_t^{1/2} \left( h_{\varphi_t}  + C \right) R_t^{1/2} \psi \rangle,\label{eq:form bound p^2}
\end{align}
where the last inequality follows from Lemma \ref{lemma: Potental}. Thus, 
\begin{align}
\|(- \Delta +1)^{1/2}  R_t^{1/2} \psi \|^2_2 &\leq C \langle  \psi, R_t^{1/2} \left( h_{\varphi_t}  -  e(\varphi_t) 
+  e(\varphi_t) + C \right) R_t^{1/2} \psi \rangle  \nonumber \\
&= C  \langle  \psi, \big( q_t + ( e(\varphi_t) + C ) R_t   \big)  \psi \rangle 
\nonumber \\
&\leq C  \langle  \psi, \big( 1 +  R_t   \big)  \psi \rangle 
\end{align}
since $e( \varphi_t) <0$ for all $|t| \leq T$ by assumption.
The gap condition then implies
\begin{align}
\label{eq: resolvent sqrt p }
\|(- \Delta +1)^{1/2} R_t^{1/2}\psi \|^2_2 \leq C  ( 1 + \Lambda^{-1}  ) .
\end{align}
In order to prove the third bound of the Lemma we calculate (with $p_t = 1 - q_t$)
\begin{align}
\dot{R}_t &=- \alpha^{-2}  p_t V_{i\varphi_t} R_t^2  - \alpha^{-2} R_t^2 V_{i \varphi_t} p_t 
+  q_t \left( \partial_t ( h_{\varphi_t} - e( \varphi_t))^{-1} \right) q_t 
\end{align}
by means of the Leibniz rule and Lemma \ref{lemma: existence minimizer}. With the resolvent identities, \eqref{eq: time derivative of the potential} and \eqref{eq:  time-derivative lowest eigenvalue}  this becomes
\begin{align}
\label{eq: derivative rho}
\dot{R}_t &=- \alpha^{-2} p_t V_{i\varphi_t} R_t^2  - \alpha^{-2} R_t^2 V_{i \varphi_t} p_t 
-  q_t ( h_{\varphi_t} - e( \varphi_t))^{-1} \left( \dot{h}_{\varphi_t} - \dot{e}( \varphi_t) \right) ( h_{\varphi_t} - e( \varphi_t))^{-1}  q_t 
\nonumber \\
&=- \alpha^{-2} p_t V_{i\varphi_t} R_t^2  - \alpha^{-2} R_t^2 V_{i \varphi_t} p_t + \alpha^{-2} R_t  \left( V_{i \varphi_t} - \langle \psi_{\varphi_t}, V_{i \varphi_t} \psi_{\varphi_t} \rangle \right) R_t .
\end{align}
Hence, Lemma \ref{lemma: Potental} leads to
\begin{align}
\| \dot{R}_t \| \leq& C \alpha^{-2} \| V_{i \varphi_t} R_t \| \; \| R_t\| + C \alpha^{-2} \| R_t \|^2 \; \|  \psi_{\varphi_t} \|_{H^1( \mathbb{R}^3)}^2 
\nonumber \\
\leq& C \alpha^{-2} \| (- \Delta +1)^{1/2} R_t \| \; \| R_t \| + \alpha^{-2} \| R_t \|^2 \; \| \psi_{\varphi_t} \|_{\Hone}^2
\nonumber \\
\leq& C \Lambda^{-3/2}\alpha^{-2}  ( 1 + \Lambda^{-1}  )^{1/2} + C \alpha^{-2} \Lambda^{-2} ,
\end{align}
for all $|t| \leq T$, where we used \eqref{eq: h1 minimizer} and the second bound of the Lemma. 
\end{proof}


\subsection{Proof of Theorem \ref{thm:adiabatic}}

In the following we denote $\widetilde{\psi}_{\varphi_t} = e^{-i \int_0^t du \; e( \varphi_u)} \psi_{\varphi_t}$. The fundamental theorem of calculus implies that
 \begin{align}
 \| \psi_t - \widetilde{\psi}_{\varphi_t} \|_{2}^2 =& -\int_0^t  ds \; \frac{d}{ds} \;  2 \; \Re \langle \psi_s, \widetilde{\psi}_{\varphi_s} \rangle 
\nonumber \\
 =&-2  \int_0^t ds \; \Re  \langle -i h_{\varphi_s} \psi_s, \widetilde{\psi}_{\varphi_s} \rangle
\nonumber \\ 
 & - 2 \int_0^t ds \; \Re \langle \psi_s, \left( -i e ( \varphi_s) + \alpha^{-2} R_s V_{i \varphi_s}\right) \widetilde{\psi}_{\varphi_s} \rangle
\nonumber \\
 =& -2 \alpha^{-2} \int_0^t ds \; \Re \langle \psi_s, \;R_s V_{i \varphi_s}  \widetilde{\psi}_{\varphi_s} \rangle,
 \label{eq:adia_1a}
 \end{align}
where we used that $h_{\varphi_s} \widetilde{\psi}_{\varphi_s} = e( \varphi_s) \widetilde{\psi}_{\varphi_s}$ and Lemma \ref{lemma: existence minimizer} to compute the derivative of the ground state $\psi_{\varphi_s}$. 
Using the Cauchy-Schwarz inequality, Lemma \ref{lemma: spectral gap} and Lemma \ref{lemma: Potental} together with Lemma \ref{lemma: resolvent} and \eqref{eq: h1 minimizer}, we obtain the inequality from Remark \ref{remark: adiabatic theorem for short times}, i.e. a bound of order $ \alpha^{-2} |t|$. In the following we shall improve this bound. 
We define $\widetilde{\psi}_s := e^{i \int_0^s d\tau \; e ( \varphi_\tau)} \psi_s$ satisfying
\begin{align}
 i\partial_s \widetilde{\psi}_s = \left( h_{\varphi_s} - e( \varphi_s) \right) \widetilde{\psi}_s
\end{align}
 and write \eqref{eq:adia_1a} as
\begin{align}
 \| \psi_t - \widetilde{\psi}_{\varphi_t} \|_{2}^2 =&-2 \alpha^{-2} \int_0^t ds \; \Re \langle q_s \widetilde{\psi}_s, \; R_s V_{i \varphi_s}  \psi_{\varphi_s} \rangle . \label{eq:adia_1}
\end{align} 
Then, we exploit that the time derivative of $\widetilde{\psi}_s$ is of order one while the time derivatives of $R_s$, $V_{i \varphi_s}$ and $\psi_{\varphi_s}$ are of order $\alpha^{-2}$; compare also with \cite[p.9]{teufel}.
We observe that
\begin{align}
i \partial_s \left( \widetilde{\psi}_s - \psi_{\varphi_s}
\right)
= \left( h_{\varphi_s} - e(\varphi_s) \right) \widetilde{\psi}_s
- i \alpha^{-2} R_s V_{i \varphi_s} \psi_{\varphi_s}
\end{align}
Hence,
\begin{align}
q_s \widetilde{\psi}_s = R_s i \partial_s \left( \widetilde{\psi}_s - \psi_{\varphi_s} \right) 
+ i \alpha^{-2} R_s^2 V_{i \varphi_s} \psi_{\varphi_s}.
\end{align}
Plugging this identity into \eqref{eq:adia_1} and using that $\Im \scp{\psi_{\varphi_s}}{V_{i \varphi_s} R_s^3 V_{i \varphi_s} \psi_{\varphi_s}} = 0$, we obtain
\begin{align}
 \| \psi_t - \widetilde{\psi}_{\varphi_t} \|_{2}^2 =& - 2 \alpha^{-2} \int_0^t ds \;  \Im  \langle   \partial_s ( \widetilde{\psi}_s - \psi_{\varphi_s} )  , \; R_s^2 V_{i \varphi_s}  \psi_{\varphi_s} \rangle .
\end{align}
Integrating by parts and using  the initial condition $\psi_0 = \psi_{\varphi_0}$  lead to
\begin{align}
 \| \psi_t - \widetilde{\psi}_{\varphi_t} \|_{2}^2 
 =&2\alpha^{-2}  \int_0^t ds\; \Im \langle   \left( \widetilde{\psi}_s - \psi_{\varphi_s} \right)  , \; \partial_s \left( R_s^2 V_{i\varphi_s} \psi_{\varphi_s}\right) \rangle 
 \nonumber \\
 & -2 \alpha^{-2} \Im \langle  \left( \widetilde{\psi}_t - \psi_{\varphi_t} \right) , \;R_t^2 V_{i\varphi_t} \psi_{\varphi_t} \rangle.
\end{align}
The Leibniz rule with \eqref{eq: derivative potential with i-varphi} and Lemma \ref{lemma: existence minimizer} lead to
\begin{subequations}
 \begin{align}
 \| \psi_t - \widetilde{\psi}_{\varphi_t} \|_{2}^2 
 =&  - 2 \alpha^{-2} \Im \langle R_t  \left(\psi_t - \widetilde{\psi}_{\varphi_t} \right)  , \;R_t V_{i \varphi_t} \widetilde{\psi}_{\varphi_t} \rangle \label{eq:adiab_2,1}\\
 &+  2 \alpha^{-2} \int_0^t ds \;  \Im \langle  \left( \psi_s - \widetilde{\psi}_{\varphi_s} \right) , \; \left( \partial_s R_s^2 \right) V_{i \varphi_s} \widetilde{\psi}_{\varphi_s}\rangle 
\label{eq:adiab_2,2} \\
 &+ 2 \alpha^{-4} \int_0^t ds \;  \Im \langle R_s  \left( \psi_s - \widetilde{\psi}_{\varphi_s} \right) , \; R_s V_{\varphi_s + \sigma_{\psi_s}} \widetilde{\psi}_{\varphi_s}\rangle \label{eq:adiab_2,4} \\
  &+ 2 \alpha^{-4} \int_0^t ds \;  \Im \langle R_s \left( \psi_s - \widetilde{\psi}_{\varphi_s} \right) , \; \left(R_s V_{i \varphi_s} \right)^2 \widetilde{\psi}_{\varphi_s}\rangle  \label{eq:adiab_2,5}  .
\end{align}
\end{subequations}
Using Lemma \ref{lemma: Potental} the first term can be estimated by
\begin{align}
\vert \eqref{eq:adiab_2,1} \vert 
&\leq C \alpha^{-2} \| R_t \|^2 \| V_{i \varphi_t} \widetilde{\psi}_{\varphi_t} \|_2   \norm{\psi_t - \widetilde{\psi}_{\varphi_t} }_2
\nonumber \\
&\leq C \alpha^{-2}  \| R_t \|^2 \| \widetilde{\psi}_{\varphi_t} \|_{H^1( \mathbb{R}^3)} \| \varphi_t \|_2  \norm{ \psi_t - \widetilde{\psi}_{\varphi_t}  }_2.
\end{align}
On the one hand Lemma \ref{lemma: well posedness LP} and \eqref{eq: h1 minimizer} show that $\| \varphi_t \|_2$ and $ \| \psi_{\varphi_t} \|_{\Hone}$ are uniformly bounded in time. On the other hand Lemma \ref{lemma: resolvent} implies that the resolvent $R_t$ is bounded for all times $|t| \leq T$, so that we obtain
\begin{align}
\vert \eqref{eq:adiab_2,1} \vert &\leq C \Lambda^{-2} \alpha^{-2}  \norm{ \psi_t - \widetilde{\psi}_{\varphi_t} }_2
\leq \frac{1}{2}  \norm{ \psi_t - \widetilde{\psi}_{\varphi_t}}_2^2 + C \Lambda^{-4} \alpha^{-4} 
, \hspace{0.5cm} \forall \;  |t| \leq T .
\end{align}
Similarly, we bound the second term by
\begin{align}
\vert \eqref{eq:adiab_2,2} \vert &\leq C \alpha^{-2} \int_0^t ds \; \| R_s \| \; \| V_{i \varphi_s} \widetilde{\psi}_{\varphi_s} \|_2 \; \| \dot{R}_s \| 
\norm{ \psi_s - \widetilde{\psi}_{\varphi_s} }_2
 \nonumber \\
&\leq C \alpha^{-4} \Lambda^{-5/2} ( 1 + \Lambda^{-1} )^{1/2}  \int_0^t ds \, \norm{ \psi_s - \widetilde{\psi}_{\varphi_s} }_2
\end{align}
for all $|t| \leq T$, using Lemma \ref{lemma: well posedness LP}, Lemma \ref{lemma: Potental} and Lemma \ref{lemma: resolvent}. 
The third term \eqref{eq:adiab_2,4}  can be bounded using $\| \sigma_{\psi_t} \|_2 \leq C \| \psi_t \|_{\Hone}^2  \leq C$ by Lemma \ref{lemma: well posedness LP}. We find
\begin{align}
\vert \eqref{eq:adiab_2,4} \vert \leq C \Lambda^{-2} \alpha^{-4}  \int_0^t ds \, \norm{ \psi_s - \widetilde{\psi}_{\varphi_s}  }_2,  \hspace{0.5cm} \forall \;  |t| \leq T .
\end{align}
Using the same ideas we estimate the last term by
\begin{align}
\vert \eqref{eq:adiab_2,5} \vert 
&\leq C \alpha^{-4} \int_0^t ds \;  \| R_s \|  \;  \| V_{i \varphi_s} R_s \|^2 
  \norm{ \psi_s - \widetilde{\psi}_{\varphi_s}  }_2
\nonumber \\
&\leq C \Lambda^{- 1} \alpha^{-4} \int_0^t ds \; \| V_{i \varphi_s} R_s \|^2
\norm{\psi_s - \widetilde{\psi}_{\varphi_s} }_2 .
\end{align}
Lemma \ref{lemma: Potental} implies that for all $\psi \in H^1( \mathbb{R}^3)$
\begin{align}
\| V_{i \varphi_s} \psi \|_2 \leq C \| \varphi_s \|_2 \; \| \psi \|_{H^1( \mathbb{R}^3)}
\end{align}
and therefore that $ \| V_{i \varphi_s} (- \Delta +1)^{-1/2} \| \leq C$ for all $|t| \leq T$ by Lemma \ref{lemma: Potental}. Hence
\begin{align}
\| V_{i \varphi_s} R_s \| \leq C \| (- \Delta  +1)^{1/2} R_s \| \leq C \Lambda^{-1/2} ( 1 + \Lambda^{-1})^{1/2},
\end{align}
for all $|t| \leq T$ where we used Lemma \ref{lemma: resolvent} for the last inequality. Thus,
\begin{align}
\vert \eqref{eq:adiab_2,5} \vert \leq C \alpha^{-4} \Lambda^{-2} ( 1 + \Lambda^{-1}) \int_0^t ds \, \norm{\psi_s - \widetilde{\psi}_{\varphi_s}  }_2, \hspace{0.5cm} \forall \;  |t| \leq T, 
\end{align}
In total, this leads to the estimate
\begin{align}
 \| \psi_t - \widetilde{\psi}_{\varphi_t} \|_{2}^2
&\leq 
C \Lambda^{-4} \alpha^{-4} 
+ C  \Lambda^{-2} \left( 1 + \Lambda^{-1}  \right)  \alpha^{-4}\int_0^t ds \, \| \psi_s - \widetilde{\psi}_{\varphi_s} \|_{2} .
\end{align}
A Gronwall type argument  then leads to
\begin{align}
\| \psi_t - \widetilde{\psi}_{\varphi_t} \|_2^2 \leq C \Lambda^{-4} \alpha^{-4} + C \Lambda^{-4} \left( 1 + \Lambda^{-1} \right)^2 \alpha^{-8} |t|^2 ,
\end{align}
completing the proof.
\end{proof}

\section{Accuracy of the Landau-Pekar equations}

\subsection{Preliminaries}

For notational convenience we define
\begin{align}
\Phi_x = \int d^3 k \; |k|^{-1} \left( e^{i k \cdot x}a_k + e^{- ik \cdot x} a_k^* \right) = \Phi_x^+ + \Phi_x^-
\end{align}
with
\begin{align}
\Phi_x^+ = \int d^3 k \; |k|^{-1}  e^{i k \cdot x} a_k, \hspace{0.5cm} \text{and}  \hspace{0.5cm} \Phi_x^- = \int d^3 k \; |k|^{-1}  e^{-i k \cdot x} a_k^*  .
\end{align}
In addition we introduce for $ f \in L^2( \mathbb{R}^3)$ the creation operator $a^*(f)$ and the annihilation operator $a(f)$ which are given by
\begin{align}
a^*(f)= \int d^3k \; f(k) a_k^*, \hspace{0.5cm} a(f) = \int d^3k \; \overline{f(k)} a_k  
\end{align}
and bounded with respect to the number of particles operator $\mathcal{N}= \int d^3k\;  a_k^* a_k$, i.e.
\begin{align}
\label{eq:bounds_a,a*}
\| a(f) \xi \| \leq \| f \|_2 \; \| \mathcal{N}^{1/2} \xi \|, \hspace{0.5cm} \| a^*(f) \xi \| \leq \| f \|_2 \| ( \mathcal{N}+\alpha^{-2})^{1/2} \xi \|
\end{align}
for all $\xi \in \mathcal{F}$.  
Moreover, recall the definition \eqref{eq: definition Weyl operator} of the Weyl operator $W(f) = e^{a^*(f) - a(f)}$.
For a time dependent function $f_t \in L^2( \mathbb{R}^3)$ the time derivative of the Weyl operator is given by
\begin{align}
\label{eq:Weyl_deriv}
\partial_t W(f_t) = \frac{\alpha^{-2}}{2} \left(  \langle f_t, \partial_t f_t\rangle - \langle \partial_t f_t,f_t\rangle \right) W(f_t) + \left( a^*(\partial_t f_t) - a(\partial_t f_t) \right) W(f_t).
\end{align}
The proof of this formula can be found in \cite[Lemma A.3]{frankgang}.

\subsection{Proof of Theorem \ref{theorem: main theorem}}

It should be noted that \eqref{eq: bound main theorem} is valid for all times which are at least of order $\alpha^2$ because both states in the inequality have norm one. To show its validity for shorter times we split the norm difference into two parts by the triangle inequality and use Remark \ref{remark: adiabatic theorem for short times} to estimate
\begin{align}
&\norm{e^{-i H_{\alpha} t} \psi_{\varphi_0} \otimes W(\alpha^2 \varphi_0) \Omega - e^{  - i \int_0^t du \, \omega(u) }   \psi_t \otimes W(\alpha^2 \varphi_t) \Omega}^2
\nonumber \\
&\quad\leq 2 \norm{e^{-i H_{\alpha} t} \psi_{\varphi_0} \otimes W(\alpha^2 \varphi_0) \Omega -  e^{  - i \int_0^t du \, \omega(u) }   \widetilde{\psi}_{\varphi_t} \otimes W(\alpha^2 \varphi_t) \Omega}^2
\nonumber \\
&\quad\quad + 2\norm{ \widetilde{\psi}_{\varphi_t} \otimes W(\alpha^2 \varphi_t) \Omega -  \psi_t \otimes W(\alpha^2 \varphi_t) \Omega}^2
\nonumber \\
&\quad\leq C \alpha^{-2} \abs{t}  + 2 \norm{e^{-i H_{\alpha} t} \psi_{\varphi_0} \otimes W(\alpha^2 \varphi_0) \Omega -  e^{  - i \int_0^t du \, \omega(u) }   \widetilde{\psi}_{\varphi_t} \otimes W(\alpha^2 \varphi_t) \Omega}^2
\end{align}
for all $|t| \leq C_\Lambda \alpha^2$ where we used the notation $\widetilde{\psi}_{\varphi_t} = e^{-i \int_0^t du \; e( \varphi_u)} \psi_{\varphi_t}$. Therefore it remains to estimate the second term.
For this, we introduce
 \begin{align}
\xi_s &= e^{   i \int_0^s du \, \omega(u)}   W^*(\alpha^2 \varphi_s) e^{-i H_{\alpha} s} \psi_{\varphi_0} \otimes W(\alpha^2 \varphi_0) \Omega
\end{align}
to shorten the notation and compute 
\begin{align*}
W^*(\alpha^2 \varphi_t) H_{\alpha} W(\alpha^2 \varphi_t)
&= h_{\varphi_t} + \norm{\varphi_t}_2^2 + \mathcal{N} + \Phi_x + a(\varphi_t) + a^*(\varphi_t) .
\end{align*}
Using \eqref{eq:Weyl_comm}, this leads to
\begin{align}
\label{eq: time derivatice xi-s}
i \partial_s \xi_s 
&= \big( i \partial_s W^*(\alpha^2 \varphi_s) \big) W(\alpha^2 \varphi_s) \xi_s
+ \Big(  W^*(\alpha^2 \varphi_s) H_{\alpha} W(\alpha^2 \varphi_s)
- \omega(s)  \Big)  \xi_s
\nonumber \\
&= \Big( h_{\varphi_s} + \Phi_x - a(\sigma_{\psi_s}) - a^*(\sigma_{\psi_s})
 +  \mathcal{N}   \Big) \xi_s .
\end{align}
Note that $\xi_0 = \psi_{\varphi_0} \otimes \Omega $. 
We apply Duhamel's formula and use the unitarity of the Weyl operator $W(\alpha^2 \varphi_t)$ to compute
\begin{align}
&\norm{e^{-i H_{\alpha} t} \psi_{\varphi_0} \otimes W(\alpha^2 \varphi_0) \Omega - e^{  - i \int_0^t du \, \omega(u) }   \widetilde{\psi}_{\varphi_t} \otimes W(\alpha^2 \varphi_t) \Omega}^2\nonumber \\
 &=  \norm{ \xi_t - \widetilde{\psi}_{\varphi_t} \otimes  \Omega}^2
= \int_0^t ds \, \partial_s \norm{\xi_s - \widetilde{\psi}_{\varphi_s} \otimes  \Omega}^2 = - 2 \Re  \int_0^t ds \,  \partial_s
\scp{ \xi_s}{\widetilde{\psi}_{\varphi_s} \otimes  \Omega} .
\end{align}
Recall that $\Phi_x = \int d^3k |k|^{-1} \left( e^{i k \cdot x} a_k + e^{- i k \cdot x} a_k^*\right)$. 
We obtain
\begin{align}
&\norm{e^{-i H_{\alpha} t} \psi_{\varphi_0} \otimes W(\alpha^2 \varphi_0) \Omega - e^{  - i \int_0^t du \, \omega(u) }  \widetilde{\psi}_{\varphi_t} \otimes W(\alpha^2 \varphi_t) \Omega}^2
\nonumber \\
&= 2 \Im \int_0^t ds \, \Big( \scp{i \partial_s \xi_s}{\widetilde{\psi}_{\varphi_s} \otimes \Omega} - \scp{\xi_s}{i \partial_s \widetilde{\psi}_{\varphi_s} \otimes \Omega}  \Big)
\nonumber \\
&= 2 \Im \int_0^t ds \, \scp{\xi_s}{\Big( h_{\varphi_s} - e(\varphi_s) + \Phi_x
- a(\sigma_{\psi_s}) - a^*(\sigma_{\psi_s})
+  \mathcal{N} \Big) \widetilde{\psi}_{\varphi_s} \otimes \Omega}
\nonumber \\
&\quad - 2 \alpha^{-2} \Re \int_0^t ds \, \scp{\xi_s}{ R_s V_{i \varphi_s} \widetilde{\psi}_{\varphi_s} \otimes \Omega}
\nonumber \\
&= 2 \Im \int_0^t ds \, \scp{\xi_s}{ \big( \Phi_x^- - a^*(\sigma_{\psi_s} ) \big) \widetilde{\psi}_{\varphi_s} \otimes \Omega}
 - 2 \alpha^{-2} \Re \int_0^t ds \, \scp{\xi_s}{ R_s V_{i \varphi_s} \widetilde{\psi}_{\varphi_s} \otimes \Omega}  .
\end{align}
Here we used the definition $R_s = q_s ( h_{\varphi_s} - e( \varphi_s) )^{-1} q_s$ and Lemma \ref{lemma: existence minimizer}.
Thus if we insert the identity $1= p_{s} + q_{s}$ and note 
that $q_s a^*( \sigma_{\psi_s}) \widetilde{\psi}_{\varphi_s} \otimes \Omega = 0$ and $q_s \widetilde{\psi}_{\varphi_s} = 0$,  we get
\begin{subequations}
\begin{align}
&\norm{e^{-i H_{\alpha} t} \psi_{\varphi_0} \otimes W(\alpha^2 \varphi_0) \Omega - e^{  - i \int_0^t du \, \omega(u) }  \widetilde{\psi}_{\varphi_t} \otimes W(\alpha^2 \varphi_t) \Omega}^2
\nonumber \\
\label{eq: gronwall c}
&\qquad  \quad = - 2 \alpha^{-2} \Re \int_0^t ds \, \scp{\xi_s}{ R_sV_{i \varphi_s} \widetilde{\psi}_{\varphi_s} \otimes \Omega} \\
\label{eq: gronwall a}
&\qquad \qquad +2 \Im \int_0^t ds \, \scp{\xi_s }{ p_s \big( \Phi_x^- - a^*(\sigma_{\psi_s} ) \big) \widetilde{\psi}_{\varphi_s} \otimes \Omega}
\\
\label{eq: gronwall b}
&\qquad \qquad + 2 \Im \int_0^t ds \, \scp{\xi_s - \widetilde{\psi}_{\varphi_s} \otimes \Omega}{ q_s \Phi_x^- \widetilde{\psi}_{\varphi_s} \otimes \Omega}.
\end{align}
\end{subequations}

\noindent
We observe that the first term \eqref{eq: gronwall c} is already of the right order, namely $\alpha^{-2} t$. To be more precise,
\begin{align}
\abs{\eqref{eq: gronwall c}}
&\leq 2 \alpha^{-2} \int_0^t ds \, \norm{\xi_s} \norm{R_s} \norm{V_{i \varphi_s}\widetilde{\psi}_{\varphi_s}}_2
\nonumber \\
&\leq 2 \alpha^{-2} \int_0^t ds \| \varphi_s \|_2 \norm{R_s} \norm{\widetilde{\psi}_{\varphi_s}}_{\Hone}
\leq C \alpha^{-2} \abs{t}
\end{align}
for all $|t| \leq C_\Lambda \alpha^2 $ where we used Lemma \ref{lemma: Potental}, Lemma \ref{lemma: resolvent}, Lemma \ref{lemma: well posedness LP} and \eqref{eq: h1 minimizer}. We estimate the remaining two terms \eqref{eq: gronwall a} and \eqref{eq: gronwall b} separately. 
 
\subsection*{The term \eqref{eq: gronwall a} }
We have
\begin{align}
\abs{\eqref{eq: gronwall a}}
&\leq    2  \int_0^t ds \,  \abs{ \scp{\xi_s  }{\int d^3k \, a_k^* \abs{k}^{-1} \Big(  \scp{\widetilde{\psi}_{\varphi_s}}{ e^{- i k \cdot} \widetilde{\psi}_{\varphi_s}} - \scp{\psi_s}{e^{- i k \cdot}\psi_s}  \Big) \widetilde{\psi}_{\varphi_s} \otimes \Omega}  }
\nonumber \\
&\leq 2 \int_0^t ds \,
\norm{ \int d^3k \, a_k^* \abs{k}^{-1} \Big(  \scp{\widetilde{\psi}_{\varphi_s}}{ e^{- i k \cdot} \widetilde{\psi}_{\varphi_s}} - \scp{\psi_s}{e^{- i k \cdot}\psi_s}  \Big) \widetilde{\psi}_{\varphi_s} \otimes \Omega}
\nonumber \\
&= 2 \int_0^t ds \,
\norm{ \int d^3k \, a_k^* \abs{k}^{-1} \Big(  \scp{\widetilde{\psi}_{\varphi_s}}{ e^{- i k \cdot} \big( \widetilde{\psi}_{\varphi_s} - \psi_s \big)} + \scp{\big( \widetilde{\psi}_{\varphi_s} - \psi_s \big)}{e^{- i k \cdot}\psi_s}  \Big) \widetilde{\psi}_{\varphi_s} \otimes \Omega} .
\end{align}
Since $\| a^*( f) \psi \otimes \Omega \| = \alpha^{-1} \| f \|_2 \; \| \psi \|$ for all $f \in L^2( \mathbb{R}^3)$, we find
\begin{align}
\abs{\eqref{eq: gronwall a}}
&\leq C \alpha^{-1} \int_0^t ds \Bigg[ \int d^3k \, \abs{k}^{-2} \Big(  \abs{\scp{\widetilde{\psi}_{\varphi_s}}{ e^{- i k \cdot} \big( \widetilde{\psi}_{\varphi_s} - \psi_s \big)}}^2 + \abs{\scp{\big( \widetilde{\psi}_{\varphi_s} - \psi_s \big)}{e^{- i k \cdot}\psi_s}}^2  \Big) \Bigg]^{1/2} .
\end{align}
With the help of $\widehat{|\cdot|^{-2}} (x) =  \pi^{-1} |x|^{-1}$ and the inequalities of Hardy-Littlewood-Sobolev and H\"older we obtain
\begin{align}
\int d^3k \, \abs{k}^{-2}  & \abs{\scp{\widetilde{\psi}_{\varphi_s}}{ e^{- i k \cdot} \big( \widetilde{\psi}_{\varphi_s} - \psi_s \big)}}^2 
\nonumber \\
&= C \int d^3x \int d^3y \, \abs{x-y}^{-1}  (\widetilde{\psi}_{\varphi_s} - \psi_s)(x) \overline{\widetilde{\psi}_{\varphi_s}(x)}  \widetilde{\psi}_{\varphi_s}(y)   \overline{ (\widetilde{\psi}_{\varphi_s} - \psi_s)(y)} 
\nonumber \\
&\leq C \norm{\widetilde{\psi}_{\varphi_s} (\widetilde{\psi}_{\varphi_s} - \psi_s)}_{6/5}^2 
\leq C \norm{\widetilde{\psi}_{\varphi_s} - \psi_s}_2^2 \norm{\widetilde{\psi}_{\varphi_s}}_3^2
\nonumber \\
& \leq C \norm{\widetilde{\psi}_{\varphi_s} - \psi_s}_2^2 \norm{\widetilde{\psi}_{\varphi_s}}_{H^1( \mathbb{R}^3)}^2 
 \leq  C \norm{\widetilde{\psi}_{\varphi_s} - \psi_s}_2^2 
\end{align}
for all $|t | \leq C_\Lambda \alpha^2$ by \eqref{eq: h1 minimizer}. Similarly,
\begin{align}
\int d^3k \, \abs{k}^{-2}  & \abs{\scp{\big( \widetilde{\psi}_{\varphi_s} - \psi_s \big)}{e^{- i k \cdot}\psi_s}}^2 
\leq C \norm{ \psi_s}_{H^1( \mathbb{R}^3)}^2 \norm{\psi_s - \widetilde{\psi}_{\varphi_s}}_2^2 \leq C \norm{\psi_s - \widetilde{\psi}_{\varphi_s}}_2^2
\end{align}
by Lemma \ref{lemma: well posedness LP}. Hence,
\begin{align}
\abs{\eqref{eq: gronwall a}}
&\leq C \alpha^{-1} \int_0^t ds  \norm{\psi_s - \widetilde{\psi}_{\varphi_s}}_2
\end{align}
for all $|t| \leq C_\Lambda \alpha^2$. Applying Theorem \ref{thm:adiabatic} leads to
\begin{align}
\abs{\eqref{eq: gronwall a}} \leq C \alpha^{-2} \abs{t}  \quad \text{for all} \, \abs{t} \leq C_\Lambda \alpha^2 .
\end{align}

\subsection*{The term \eqref{eq: gronwall b}}

In order to continue we note that \cite[Theorem X.71]{reedsimon}, whose assumptions can easily shown to be satisfied by Lemma \ref{lemma: Potental}, guarantees the existence of a two parameter group $U_{h}(s;\tau)$ on $L^2(\mathbb{R}^3)$ such that 
\begin{align}
\frac{d}{ds} U_h(s; \tau) \psi = - i h_{\varphi_s} U_h(s; \tau)\psi ,
\quad U_h(\tau; \tau) \psi = \psi 
\quad \text{for all} \, \psi \in H^1(\mathbb{R}^3).
\end{align}
Moreover, we define
\begin{align}
\widetilde{U}_{h}(s; \tau) &= e^{ i \int_{\tau}^s du \, e(\varphi_u)} U_h(s; \tau) .
\end{align}
We then have for all $s \in \mathbb{R}$
\begin{align}
\frac{d}{ds} \widetilde{U}_{h}^*(s; \tau)  = \widetilde{U}_{h}^*(s; \tau) i \big(h_{\varphi_s} - e(\varphi_s) \big)
\end{align}
and 
\begin{align}
\label{eq:parts2}
\widetilde{U}^*_{h}(s; \tau) q_s f_s
&= - i \frac{d}{ds}\left[  \widetilde{U}_{h}^*(s; \tau) R_s  f_s\right] + i \widetilde{U}_{h}^*(s; \tau)\dot{R_s} f_s + i  \widetilde{U}_{h}^*(s; \tau) R_s  \partial_s f_s,
\end{align}
for $f_s \in L^2( \mathbb{R}^3)$.
This allows us to express 
\begin{align}
\eqref{eq: gronwall b}
&= 2 \Im \int_0^t ds \, \scp{\xi_s - \widetilde{\psi}_{\varphi_s} \otimes \Omega}{ q_s \Phi_x^- \widetilde{\psi}_{\varphi_s} \otimes \Omega}
\nonumber \\
&= 2 \Im \int_0^t ds \, \scp{ U_h^*(s;0) \big( \xi_s - \widetilde{\psi}_{\varphi_s} \otimes \Omega \big)}{ \widetilde{U}_{h}^*(s; 0) q_s  \Phi_x^- \psi_{\varphi_s} \otimes \Omega}
\end{align}
by three integrals which contain a derivative with respect to the time variable.
Note that we absorbed the phase factor of $\widetilde{\psi}_{\varphi_s}$ in the dynamics $\widetilde{U}_{h}^*(s;0)$. Thus, 
\begin{align}
\abs{\eqref{eq: gronwall b}}
&\leq 2 \abs{\int_0^t ds \, \scp{U_h^*(s;0) \big( \xi_s - \widetilde{\psi}_{\varphi_s} \otimes \Omega \big)}{ \frac{d}{ds}  \left[ {U}_{h}^*(s;0)R_s \Phi_x^- \psi_{\varphi_s} \otimes \Omega\right] }}
\nonumber \\
&\quad 
+ 2 \abs{\int_0^t ds \, \scp{U_h^*(s;0)  \big( \xi_s - \widetilde{\psi}_{\varphi_s} \otimes \Omega \big)}{ \widetilde{U}_{h}^*(s;0)\dot{R}_s \Phi_x^- \psi_{\varphi_s} \otimes \Omega}} 
\nonumber \\
&\quad 
+ 2 \abs{\int_0^t ds \, \scp{U_h^*(s;0)  \big( \xi_s - \widetilde{\psi}_{\varphi_s} \otimes \Omega \big)}{ \widetilde{U}_{h}^*(s;0) R_s \Phi_x^- \partial_s  \psi_{\varphi_s} \otimes \Omega}}  .
\end{align}
In the first term we integrate by parts and we use $\xi_0 = \widetilde{\psi}_{\varphi_0} \otimes \Omega$. We find
\begin{align}
\abs{\eqref{eq: gronwall b}}
&\leq 2 \abs{ \scp{U_h^*(t;0) \big( \xi_t - \widetilde{\psi}_{\varphi_t} \otimes \Omega \big) }{U_h^*(t;0) R_t \Phi_x^- \widetilde{\psi}_{\varphi_t} \otimes \Omega}}
\nonumber \\
&\quad + 2
\abs{ \int_0^t ds \, \scp{\frac{d}{ds}  \left[ U_h^*(s;0) \big( \xi_s - \widetilde{\psi}_{\varphi_s} \otimes \Omega \big) \right] }{U_h^*(s;0) R_s \Phi_x^- \widetilde{\psi}_{\varphi_s} \otimes \Omega}}
\nonumber \\
&\quad + 2
\abs{\int_0^t ds \, \scp{\big( \xi_s - \widetilde{\psi}_{\varphi_s} \otimes \Omega \big)}{\dot{R}_s \Phi_x^- \widetilde{\psi}_{\varphi_s} \otimes \Omega}}
\nonumber \\
&\quad 
+ 2 \abs{\int_0^t ds \, \scp{  \xi_s }{ R_s  \Phi_x^- \partial_s  \psi_{\varphi_s} \otimes \Omega}}  . \label{eq:92_1}
\end{align}
In order to compute the time derivative occurring in the second summand, we use \eqref{eq: time derivatice xi-s} and the notation
\begin{align}
\label{eq: definition of delta H}
\delta H_s  =  \Phi_x - a(\sigma_{\psi_s}) - a^*(\sigma_{\psi_s}) +  \mathcal{N}
\end{align}
and get
\begin{align}
\label{eq:hf1}
\frac{d}{ds} \left[ U_h^*(s;0) \big( \xi_s - \widetilde{\psi}_{\varphi_s} \otimes \Omega \big)\right]
&= - i U_h^*(s;0) \delta H_s \,  \xi_s - \alpha^{-2} U_h^*(s;0) R_s V_{i \varphi_s} \widetilde{\psi}_{\varphi_s} \otimes \Omega
\end{align}
as well as
\begin{align}
\label{eq:hf2}
U_h^*(t;0) \big( \xi_t - \widetilde{\psi}_{\varphi_t} \otimes \Omega \big)
&= - i \int_0^t ds \, U_h^*(s;0)   \delta H_s \, \xi_s
- \alpha^{-2} \int_0^t ds \, U_h^*(s;0) R_s V_{i \varphi_s} \widetilde{\psi}_{\varphi_s} \otimes \Omega.
\end{align}
Applying \eqref{eq:hf2} to the first term of the r.h.s. of \eqref{eq:92_1} and \eqref{eq:hf1}  to the second, we obtain
\begin{subequations}
\begin{align}
\label{eq: gronwall b1}
\abs{\eqref{eq: gronwall b}}
&\leq 2  \abs{\int_0^t ds \, \scp{\delta H_s \, \xi_s}{R_s \Phi_x^- \widetilde{\psi}_{\varphi_s} \otimes \Omega}}
\\
\label{eq: gronwall b2}
&\quad + 2  \abs{\int_0^t ds \, \scp{\delta H_s \, \xi_s}{ U_h^*(t;s) R_t \Phi_x^- \widetilde{\psi}_{\varphi_t} \otimes \Omega}}
\\
\label{eq: gronwall b3}
&\quad + 2 \alpha^{-2}
\abs{\int_0^t ds \, \scp{R_s V_{i \varphi_s} \widetilde{\psi}_{\varphi_s} \otimes \Omega}{ R_s \Phi_x^- \widetilde{\psi}_{\varphi_s} \otimes \Omega}}
\\
\label{eq: gronwall b4}
&\quad + 2 \alpha^{-2}
\abs{\int_0^t ds \, \scp{R_s V_{i \varphi_s} \widetilde{\psi}_{\varphi_s} \otimes \Omega}{U_h^*(t;s) R_t \Phi_x^- \widetilde{\psi}_{\varphi_t} \otimes \Omega}}
\\
\label{eq: gronwall b5}
&\quad + 
2 \abs{\int_0^t ds \, \scp{\big( \xi_s - \widetilde{\psi}_{\varphi_s} \otimes \Omega \big)}{\dot{R}_s \Phi_x^- \widetilde{\psi}_{\varphi_s} \otimes \Omega}}
\\
\label{eq: gronwall b6}
& \quad +
2 \abs{\int_0^t ds \, \scp{ \xi_s }{ R_s  \Phi_x^- \partial_s  \psi_{\varphi_s} \otimes \Omega}}   .
\end{align}
\end{subequations}

\paragraph{The term \eqref{eq: gronwall b1}:}  According to the definition of $\delta H_s$, we decompose  \eqref{eq: gronwall b1} as

\begin{subequations}
\begin{align}
\label{eq: gronwall b1a}
\eqref{eq: gronwall b1}
&\leq 2 \int_0^t ds \, \abs{\scp{\xi_s}{  \mathcal{N} R_s 
\Phi_x^- \widetilde{\psi}_{\varphi_s} \otimes \Omega }}
 \\
\label{eq: gronwall b1b}
&\quad + 2 \int_0^t ds \, \abs{\scp{\xi_s}{\big( a(\sigma_{\psi_s}) + a^*(\sigma_{\psi_s}) \big)  R_s \Phi_x^- \widetilde{\psi}_{\varphi_s} \otimes \Omega }}
 \\
\label{eq: gronwall b1c}
&\quad + 2 \int_0^t ds \, \abs{\scp{\xi_s}{ \Phi_x R_s \Phi_x^- \widetilde{\psi}_{\varphi_s} \otimes \Omega }}.
\end{align}
\end{subequations}
We notice that 
$\big[ \mathcal{N} , R_s \big] =0$ and 
that $\mathcal{N} \Psi = \alpha^{-2} \Psi$ if $\Psi \in \mathcal{H}$ is a one-phonon state
and write the first line as
\begin{align}
\eqref{eq: gronwall b1a}
&= 2 \alpha^{-2} \int_0^t ds \,  \abs{  \scp{\xi_s}{R_s  \Phi_x^- \widetilde{\psi}_{\varphi_s} \otimes \Omega }}
\nonumber \\
&=  2 \alpha^{-2} \int_0^t ds \,  \abs{  \scp{\xi_s}{R_s  \int d^3k \, \abs{k}^{-1} e^{- ikx}  a_k^* \widetilde{\psi}_{\varphi_s} \otimes \Omega }}.
\end{align}
By means of  Lemma \ref{lemma: resolvent} and Lemma \ref{lemma: frankgang 3.1} this becomes
\begin{align}
\eqref{eq: gronwall b1a}
&\leq C \alpha^{-3}  \int_0^t ds \, \norm{\xi_s} \norm{R_s(- \Delta  + 1)^{1/2}} \norm{\psi_{\varphi_s}}_{2}
\leq C \alpha^{-3} \abs{t} 
\end{align}
for all $|t| \leq C_\Lambda \alpha^2$. In a similar way, we calculate $\big[ a(\sigma_{\psi_s}), a_k^*  \big] = \alpha^{-2} \overline{\sigma_{\psi_s}(k)}$ for all $k \in \mathbb{R}^3$ and estimate
\begin{align}
\eqref{eq: gronwall b1b}
&= 2 \int_0^t ds \, \abs{\scp{\xi_s}{\big( a(\sigma_{\psi_s}) + a^*(\sigma_{\psi_s}) \big) R_s \int d^3k \, \abs{k}^{-1} e^{- i k \cdot x} a_k^* \, \widetilde{\psi}_{\varphi_s} \otimes \Omega }}
\nonumber \\
&\leq 2 \int_0^t ds \, \abs{\scp{\xi_s}{R_s \int d^3k \, \abs{k}^{-1} e^{- i k \cdot x} a^*(\sigma_{\psi_s}) a_k^* \widetilde{\psi}_{\varphi_s} \otimes \Omega }}
\nonumber\\
&\quad +  2  \alpha^{-2} \int_0^t ds \, \abs{\scp{\xi_s}{ R_s \int d^3k \, \abs{k}^{-1} e^{- i k \cdot x} \overline{\sigma_{\psi_s}(k)} \, \widetilde{\psi}_{\varphi_s} \otimes \Omega }} .
\end{align}
Applying Lemma \ref{lemma: frankgang 3.1}, Lemma \ref{lemma: Potental} and Lemma \ref{lemma: resolvent} to the first line and using the same arguments as in  Lemma \ref{lemma: frankgang 3.1} for the second line this becomes 
\begin{align}
\eqref{eq: gronwall b1b}
&\leq C \alpha^{-2} \int_0^t ds \, \norm{\xi_s} \norm{R_s (- \Delta  +1)^{1/2}} \norm{\sigma_{\psi_s}}_{2} \norm{\psi_{\varphi_s}}_{2}
\nonumber \\
&\leq C \alpha^{-2} \abs{t}
\quad \text{for all} \, \abs{t} \leq C_\Lambda \alpha^2 .
\end{align}
Since $\Phi_x = \Phi_x^{+} + \Phi_x^{-}$ we have
\begin{subequations}
\begin{align}
\eqref{eq: gronwall b1c}
&= 2 \int_0^t ds \, \abs{\scp{\xi_s}{\Phi_x R_s \Phi_x^- \widetilde{\psi}_{\varphi_s} \otimes \Omega}}
\nonumber \\
\label{eq: gronwall b1c1}
&\leq 2  \int_0^t ds \, \abs{\scp{\xi_s}{\Phi_x^+ R_s \Phi_x^- \widetilde{\psi}_{\varphi_s} \otimes \Omega}}
\\
\label{eq: gronwall b1c2}
&\quad + 2 \int_0^t ds \, \abs{ \scp{\Phi_x^+ \xi_s}{ R_s \Phi_x^- \widetilde{\psi}_{\varphi_s} \otimes \Omega}}.
\end{align}
\end{subequations}
Making use of Lemma \ref{lemma: frank schlein lemma 10} the first line can be estimated by
\begin{align}
\eqref{eq: gronwall b1c1}
&\leq C \int_0^t ds \, \norm{(- \Delta  + 1)^{1/2} \mathcal{N}^{1/2} R_s \Phi_x^- \widetilde{\psi}_{\varphi_s} \otimes \Omega}
\nonumber \\
&= C \int_0^t ds \, \norm{(- \Delta  + 1)^{1/2} R_s \mathcal{N}^{1/2}   \Phi_x^- \widetilde{\psi}_{\varphi_s} \otimes \Omega} .
\end{align}
Since $\| (- \Delta +1)^{1/2} R_s (- \Delta  +1)^{1/2} \| \leq C$ for all $|t| \leq C_\Lambda \alpha^2$ by Lemma \ref{lemma: resolvent} and $\mathcal{N}^{1/2} \Psi = \alpha^{-1} \Psi$ if $\Psi \in \mathcal{H}$ is a one-phonon state, we find 
\begin{align}
\eqref{eq: gronwall b1c1}
&\leq  C \alpha^{-1} \int_0^t ds \,  \norm{  (- \Delta  + 1)^{-1/2}  \Phi_x^- \widetilde{\psi}_{\varphi_s} \otimes \Omega}
\end{align} 
for all $\abs{t} \leq C_\Lambda \alpha^2.$
With Lemma \ref{lemma: frankgang 3.1} we arrive at
\begin{align}
\eqref{eq: gronwall b1c1}
&\leq C \alpha^{-2} \int_0^t ds \, \norm{\psi_{\varphi_s}}_{2}
\leq C \alpha^{-2} \abs{t}
\end{align}
for all $\abs{t} <C_\Lambda \alpha^2.$
In similar fashion we use  Lemma \ref{lemma: frank schlein lemma 10},  Lemma \ref{lemma: frankgang 3.1}  and $\mathcal{N} R_s \Phi_x^- \widetilde{\psi}_{\varphi_s} \otimes \Omega = \alpha^{-2} R_s \Phi_x^- \widetilde{\psi}_{\varphi_s} \otimes \Omega$ to estimate
\begin{align}
\label{eq: estimates on the difficult term 1}
\eqref{eq: gronwall b1c2}
&=  2 \int_0^t ds \, \abs{ \scp{\Phi_x^+ \xi_s}{ R_s \Phi_x^- \widetilde{\psi}_{\varphi_s} \otimes \Omega}}
\nonumber \\
&=  2 \int_0^t ds \, \abs{ \scp{ (\mathcal{N} + \alpha^{-2})^{-1/2}  \Phi_x^+ \xi_s}{ (\mathcal{N} + \alpha^{-2})^{1/2} R_s \Phi_x^- \widetilde{\psi}_{\varphi_s} \otimes \Omega}}
\nonumber \\
&\leq 2 \int_0^t ds \, \norm{ (\mathcal{N} + \alpha^{-2})^{-1/2}  \Phi_x^+ \xi_s}
\norm{(\mathcal{N} + \alpha^{-2})^{1/2} R_s \Phi_x^- \widetilde{\psi}_{\varphi_s} \otimes \Omega}
\nonumber \\
&\leq C \alpha^{-1}
\int_0^t ds \, \norm{(- \Delta  +1)^{1/2} \xi_s} \norm{R_s \Phi_x^- \widetilde{\psi}_{\varphi_s} \otimes \Omega}
\nonumber \\
&\leq C \alpha^{-1}
\int_0^t ds \, \norm{(- \Delta  +1)^{1/2} \xi_s}
\norm{R_s (- \Delta  +1)^{1/2}}
\norm{(- \Delta   + 1)^{-1/2} \Phi_x^- \widetilde{\psi}_{\varphi_s} \otimes \Omega}
\nonumber \\
&\leq C \alpha^{-2} \int_0^t ds \, \norm{(- \Delta  +1)^{1/2} \xi_s} \norm{\psi_{\varphi_s}}_{2}
\nonumber \\
&= C   \alpha^{-2} \int_0^t ds \, \norm{(- \Delta  +1)^{1/2} e^{- i H_{\alpha} s} \psi_{\varphi_0} \otimes W(\alpha^2 \varphi_0) \Omega}
\end{align}
for all  $\abs{t} < C_{\Lambda} \alpha^2$. 
Thus, if we now use $- \Delta  + 1 \leq C (H_{\alpha} + C)$ (see Lemma \ref{lemma: bound for the Froehlich Hamiltonian})  this becomes using the properties \eqref{eq:Weyl_comm} of the Weyl operators 
\begin{align}
\label{eq: estimates on the difficult term 2}
\eqref{eq: gronwall b1c2}
&\leq C \alpha^{-2} \int_0^t ds \, \norm{(H_\alpha + C )^{1/2} e^{- i H_\alpha s} \psi_{\varphi_0} \otimes W( \alpha^2 \varphi_0)  \Omega}\notag \\
&=  C \alpha^{-2} \int_0^t ds \, \norm{(H_\alpha + C)^{1/2}  \psi_{\varphi_0} \otimes W( \alpha^2 \varphi_0 ) \Omega}
\nonumber \\
&= C \alpha^{-2} \int_0^t ds \, \big(  e(\varphi_0) + \norm{\varphi_0}^2 + 1  \big)^{1/2}
= C \alpha^{-2} \abs{t}
\end{align}
for all $\abs{t} < C_\Lambda\alpha^2$.
In total, we obtain $\eqref{eq: gronwall b1c} \leq C \alpha^{-2} \abs{t} $ and hence
$\eqref{eq: gronwall b1} \leq C \alpha^{-2} \abs{t}$ for all $\abs{t} < C_\Lambda\alpha^2$.

\paragraph{The term \eqref{eq: gronwall b2}:}

For the next estimate, we recall the notation \eqref{eq: definition of delta H} to write \eqref{eq: gronwall b2} as
\begin{align}
\eqref{eq: gronwall b2}
&\leq 
2  \int_0^t ds \, \abs{ \scp{ \xi_s}{ \mathcal{N} U_h^*(t;s) R_t \Phi_x^- \widetilde{\psi}_{\varphi_t} \otimes \Omega}}
\nonumber \\
&\quad + 2  \int_0^t ds \, \abs{ \scp{  \xi_s}{ \big(  a(\sigma_{\psi_s}) + a^*(\sigma_{\psi_s} ) \big) U_h^*(t;s)  R_t \Phi_x^- \widetilde{\psi}_{\varphi_t} \otimes \Omega}}
\nonumber \\
&\quad + 2  \int_0^t ds \, \abs{ \scp{\xi_s}{ \Phi_x U_h^*(t;s) R_t \Phi_x^- \widetilde{\psi}_{\varphi_t} \otimes \Omega}} .
\end{align}
Using
$[\mathcal{N}, U_h^*(t;s)] =  [ a(\sigma_{\psi_s}), U_h^*(t;s)] = [a^*(\sigma_{\psi_s}), U_h^*(t;s)] = 0$
allows us to estimate the first two lines in exactly the same way as \eqref{eq: gronwall b1a} and \eqref{eq: gronwall b1b}
and leaves us with
\begin{align}
\eqref{eq: gronwall b2} 
 &\leq C \alpha^{-2} \abs{t} 
 +  2  \int_0^t ds \, \abs{ \scp{\xi_s}{ \Phi_x U_h^*(t;s) R_t \Phi_x^- \widetilde{\psi}_{\varphi_t} \otimes \Omega}}
\end{align}
for all $\abs{t} < C_\Lambda \alpha^2$.
The difficulty of this term is the fact that the operators $\Phi_x$ and $U_h^*(t;s)$ do not commute. 
Nevertheless, we can use $\Phi_x = \Phi_x^+ + \Phi_x^-$ to get
\begin{subequations}
\begin{align}
\eqref{eq: gronwall b2} 
 &\leq C \alpha^{-2} \abs{t} 
 +  2  \int_0^t ds \, \abs{\scp{\Phi_x^+ \xi_s}{  U_h^*(t;s) R_t \Phi_x^- \widetilde{\psi}_{\varphi_t} \otimes \Omega}}
\\
\label{eq: gronwall b2b} 
&\quad + 2  \int_0^t ds \, \abs{ \scp{\xi_s}{ \Phi_x^+ U_h^*(t;s) R_t \Phi_x^- \widetilde{\psi}_{\varphi_t} \otimes \Omega}}.
\end{align}
\end{subequations}
Using the same estimates as in \eqref{eq: estimates on the difficult term 1} and  \eqref{eq: estimates on the difficult term 2} we bound the first integral by
\begin{align}
2  \int_0^t ds \, &\abs{\scp{\Phi_x^+ \xi_s}{  U_h^*(t;s)  R_t \Phi_x^- \tilde{\psi}_{\varphi_t} \otimes \Omega}}
\nonumber \\
&= 2  \int_0^t ds \, \abs{\scp{(\mathcal{N} + \alpha^{-2})^{-1/2} \Phi_x^+ \xi_s}{  U_h^*(t;s)  (\mathcal{N} + \alpha^{-2})^{1/2} R_t \Phi_x^- \tilde{\psi}_{\varphi_t} \otimes \Omega}}
\nonumber \\
&\leq 2 \int_0^t ds \, \norm{(\mathcal{N} + \alpha^{-2})^{-1/2} \Phi_x^+ \xi_s}
\norm{(\mathcal{N} + \alpha^{-2})^{1/2} R_t \Phi_x^- \tilde{\psi}_{\varphi_t} \otimes \Omega}
\nonumber \\
&\leq C \alpha^{-2} \abs{t} 
\quad \text{for all} \, \abs{t} \leq C_{\Lambda} \alpha^2.
\end{align}
For the second term Lemma \ref{lemma: frank schlein lemma 10} and $U_h^*(t;s) = U_h(s;t)$ imply
\begin{align}
\eqref{eq: gronwall b2b} 
&\leq C  \int_0^t ds \, \norm{ (- \Delta  + 1)^{1/2} \mathcal{N}^{1/2} U_h(s;t) R_t \Phi_x^- \tilde{\psi}_{\varphi_t} \otimes \Omega} .
\end{align}
It follows from Lemma \ref{lemma: Potental} and \eqref{eq: time derivative of the potential} that for $\xi \in L^2( \mathbb{R}^3) \otimes \mathcal{F}$ 
\begin{align}
\langle \xi, U^*_h(s;\tau) (- \Delta +1) U_h(s;\tau) \xi \rangle \leq& C \langle \xi, U^*_h(s;\tau) ( h_{\varphi_s} + 1) U_h(s;\tau) \xi \rangle \notag \\
=& C \langle \xi, ( h_{\varphi_{\tau}} + 1) \xi \rangle  - \alpha^{-2} \int_\tau^s d\tau ' \; \langle \xi, U^*_h(\tau ';\tau) V_{i \varphi_{\tau '}} U_h(\tau ';\tau) \xi \rangle\notag \\
\leq& C \langle \xi, (- \Delta  +1)\xi \rangle  + \alpha^{-2} \int_\tau^s d\tau ' \; \langle \xi, U^*_h(\tau ';\tau) (- \Delta  +1) U_h(\tau ';\tau) \xi \rangle.
\end{align}
The Gronwall inequality yields 
\begin{align}
\| (- \Delta +1)^{1/2}  U_h(s;\tau )   \xi \| \leq C e^{ \alpha^{-2} |s - \tau|} \| (- \Delta  +1)^{1/2}\xi \| \leq C\| (- \Delta +1)^{1/2}\xi \| 
\end{align}
for all $|s - \tau| \leq C_\Lambda \alpha^2$. Thus
\begin{align}
\eqref{eq: gronwall b2b} 
&\leq C  \int_0^t ds \, \norm{ (- \Delta  + 1)^{1/2} \mathcal{N}^{1/2} R_t \Phi_x^- \widetilde{\psi}_{\varphi_t} \otimes \Omega} \leq C \alpha^{-2} |t|
\end{align}
for all $|t| \leq C_\Lambda \alpha^{-2}$, where we concluded by Lemma \ref{lemma: frankgang 3.1} and Lemma \ref{lemma: resolvent} as for the term \eqref{eq: gronwall b1c1}.

\paragraph{The terms \eqref{eq: gronwall b3} and \eqref{eq: gronwall b4}:}

With the help of Lemma \ref{lemma: Potental}, Lemma \ref{lemma: frankgang 3.1}, Lemma \ref{lemma: resolvent} and \eqref{eq: h1 minimizer} one obtains
\begin{align}
\eqref{eq: gronwall b3}
&= 2 \alpha^{-2} \abs{  \int_0^t ds \, \scp{R_s V_{i \varphi_s} \widetilde{\psi}_{\varphi_s} \otimes \Omega}{ R_s \int d^3k \, \abs{k}^{-1}  e^{- ik \cdot x} a_k^* \widetilde{\psi}_{\varphi_s} \otimes \Omega}}
\nonumber \\
&\leq 2 \alpha^{-2} \int_0^t ds \, 
\norm{R_s} \norm{\psi_{\varphi_s}}_{H^1( \mathbb{R}^3)} \norm{R_s (- \Delta  +1)^{1/2}}  \norm{(- \Delta  + 1)^{-1/2} \int d^3k \, \abs{k}^{-1} e^{- ikx} a_k^* \widetilde{\psi}_{\varphi_s} \otimes \Omega}
\nonumber \\
&\leq C \alpha^{-3} \abs{t}
\end{align}
and 
$
\eqref{eq: gronwall b4} \leq C \alpha^{-3} \abs{t}
$
for all $\abs{t} < C_\Lambda\alpha^2$.

\paragraph{The term \eqref{eq: gronwall b5}:}
Applying Lemma \ref{lemma: frankgang 3.1}   once more we estimate 

\begin{align}
\eqref{eq: gronwall b5}
&= 2  \abs{ \int_0^t ds \, 
\scp{\big( \xi_s - \widetilde{\psi}_{\varphi_s} \otimes \Omega \big)}{ \dot{R}_s
\int d^3k \, \abs{k}^{-1}  e^{- ik \cdot x} a_k^* \,  \widetilde{\psi}_{\varphi_s} \otimes \Omega}}
\nonumber \\
&\leq 4 \int_0^t ds \, \norm{\dot{R}_s\big( - \Delta  + 1 \big)^{1/2}}
\norm{( - \Delta  + 1 )^{-1/2} \int d^3k \, \abs{k}^{-1}  e^{- ik \cdot x} a_k^* \,  \widetilde{\psi}_{\varphi_s} \otimes \Omega}  .
\end{align}
From \eqref{eq: h1 minimizer},  \eqref{eq: derivative rho} and Lemma  \ref{lemma: resolvent} we get
\begin{align}
\eqref{eq: gronwall b5}
&\leq C \alpha^{-3} \abs{t} 
\quad  \text{for all} \, \abs{t} < C_\Lambda\alpha^2.
\end{align}

\paragraph{The term \eqref{eq: gronwall b6}:} With the help of Lemma \ref{lemma: existence minimizer} and Lemma \ref{lemma: Potental} we get

\begin{align}
\eqref{eq: gronwall b6} &= 2 \alpha^{-2}\abs{\int_0^t ds \, \scp{ \xi_s }{R_s  \Phi_x^- R_sV_{i \varphi_s} \psi_{\varphi_s} \otimes \Omega}}  
\nonumber \\
&\leq 2 \alpha^{-2}
\int_0^t ds \,
\norm{R_s(- \Delta  +1)^{1/2}}
\norm{(- \Delta  +1 )^{-1/2} \Phi_x^- R_s V_{i \varphi_s} \psi_{\varphi_s} \otimes \Omega} 
\nonumber \\
&\leq C \alpha^{-3} \int_0^t ds \; \norm{R_s(- \Delta  +1)^{1/2}} \| R_sV_{i \varphi_s} \| \; \| \psi_{\varphi_s} \|_2 \leq C \alpha^{-3 } |t| .
\end{align}
Here we used again Lemma \ref{lemma: frankgang 3.1} and Lemma \ref{lemma: resolvent}. 
In total, we obtain
\begin{align}
\abs{\eqref{eq: gronwall b}}
&\leq C \alpha^{-2} \abs{t} 
\quad \text{for all} \, \abs{t} <  C_\Lambda \alpha^2.
\end{align}
Summing up, we have shown that
\begin{align}
&\norm{e^{-i H_{\alpha} t} \psi_{\varphi_0} \otimes W(\alpha^2 \varphi_0) \Omega - e^{  - i \int_0^t du \, \omega(u) }  \widetilde{\psi}_{\varphi_t} \otimes W(\alpha^2 \varphi_t) \Omega}^2 \leq C \alpha^{-2 } |t|  ,
\end{align}
for all $|t| \leq C_\Lambda \alpha^2$.

\appendix

\section{Auxiliary estimates}

\begin{lemma}
\label{lemma: frankgang 3.1}
There exists a constant $C>0$ such that for all $u \in L^2(\mathbb{R}^3)$ and $f \in L^2(\mathbb{R}^3)$
\begin{align}
\norm{( - \Delta + 1)^{-1/2} \int d^3k \, \abs{k}^{-1} e^{- ik \cdot x} a_k^* u \otimes \Omega  } \leq C \alpha^{-1} \norm{u}_{2} ,
\\
\norm{(- \Delta + 1)^{-1/2} \int d^3k \, \abs{k}^{-1} e^{- ik \cdot x}  a^*(f) a_k^* u \otimes \Omega  } \leq C \alpha^{-2} \norm{u}_{2} \norm{f}_2 .
\end{align}
\end{lemma}

\begin{proof}
The commutation relations imply
\begin{align}
\| ( - \Delta + 1)^{-1/2} &\int d^3k \; |k|^{-1} e^{- ik \cdot x} a_k^* u \otimes \Omega \|^2 \notag\\
=& \int d^3k \int d^3k' \; |k|^{-1} |k'|^{-1} \langle e^{- ik \cdot x} a_k^* u \otimes \Omega , ( - \Delta + 1)^{-1} e^{- ik' \cdot x} a_{k'}^* u \otimes \Omega \rangle \notag\\
=& \alpha^{-2} \int d^3k \;  |k|^{-2} \langle e^{- ik \cdot x} u , (  - \Delta  +1)^{-1} e^{- ik \cdot x} u \rangle \notag\\
=& \alpha^{-2} \int d^3k \;  |k|^{-2} \langle u , ( (- i \nabla -k)^2 +1)^{-1}  u \rangle\notag \\
=&  \alpha^{-2} \int d^3p\;  | \hat{u}(p)|^2 \int d^3k \;   \frac{1}{((p-k)^2+1) |k|^2} .
\end{align}
Since $| \cdot |^{-2}$ and $(|\cdot|^2 +1)^{-1}$ are radial symmetric and decreasing functions we have
\begin{align}
\sup_{p \in \mathbb{R}^3} \int d^3k \;   \frac{1}{((p-k)^2+1) |k|^2}
&= \int d^3k \;   \frac{1}{(k^2+1) |k|^2} < \infty
\end{align}
by the rearrangement inequality.
Hence,
\begin{align}
\| (- \Delta +  1)^{-1/2} &\int d^3 |k|^{-1} e^{- ik \cdot x} a_k^* u \otimes \Omega \|^2 \leq C \alpha^{-2} \int d^3p \; | \hat{u}(p) |^2 = C \alpha^{-2} \| u \|_2 .
\end{align}
The second bound of the Lemma follows from the first one and the bounds \eqref{eq:bounds_a,a*} for the creation and annihilation operators.
\end{proof}

\begin{lemma}[Lemma 4, Lemma 10 in \cite{frankschlein}]
\label{lemma: frank schlein lemma 10}
Let $\Phi_x^+ = \int d^3k \, \abs{k}^{-1}  e^{ik \cdot x} a_k$ and $\mathcal{N} = \int d^3k \, a^*_k a_k$. Then
\begin{align}
\norm{\Phi_x^+ \Psi} &\leq  C \norm{(- \Delta  +1)^{1/2} \mathcal{N}^{1/2} \Psi}
\quad \text{and} \quad
\norm{(\mathcal{N} +\alpha^{-2})^{-1/2} \Phi_x^+ \Psi} \leq C \norm{(- \Delta  + 1)^{1/2} \Psi}.
\end{align}
\end{lemma}
\begin{proof}
We split the operator $\Phi_x^+ = \Phi_x^{+, >} +\Phi_x^{+, <}  $, where 
\begin{align}
\Phi_x^{+, >} = \int_{|k|> \kappa} \frac{d^3k}{|k|} e^{ik \cdot x}  a_k, \hspace{0.5cm}\Phi_x^{+, <} = \int_{|k|< \kappa} \frac{d^3k}{|k|}  e^{ik \cdot x} a_k
\end{align}
for a constant $\kappa >0$ of order one. Then, we deduce from \cite[Lemma 10]{frankschlein} that 
\begin{align}
\norm{\Phi_x^{+,>} \Psi} &\leq  C \norm{(- \Delta  +1)^{1/2} \mathcal{N}^{1/2} \Psi}
\quad \text{and} \quad
\norm{(\mathcal{N} +\alpha^{-2})^{-1/2} \Phi_x^{+,>} \Psi} \leq C \norm{(- \Delta  + 1)^{1/2} \Psi} ,
\end{align}
where the constant $C>0$ depends only on $\kappa$. Since the function $f_< : \mathbb{R}^3 \rightarrow \mathbb{R}$ given through  $f_<(k) =|k| \chi_{|k| \leq \kappa }$ is in $L^2( \mathbb{R}^3)$, \cite[Lemma 4]{frankschlein} implies
\begin{align}
\norm{\Phi_x^{+,<} \Psi} &\leq  C \norm{\mathcal{N}^{1/2} \Psi}
\quad \text{and} \quad
\norm{(\mathcal{N} +\alpha^{-2})^{-1/2} \Phi_x^{+,<} \Psi} \leq C \norm{ \Psi}  .
\end{align}
for a constant $C>0$ depending only on $\kappa$. 
\end{proof}
%
%
\begin{lemma}[\cite{frankschlein}, p.7]
\label{lemma: bound for the Froehlich Hamiltonian}
Let $\alpha_0 >0$. Let $H_\alpha$ denote the Fr\"ohlich Hamiltonian defined in \eqref{eq: Froehlich Hamiltonian} and $\varepsilon >0$. There exists a constant $C_\varepsilon$ (depending on $\alpha_0$) , such that
\begin{align}
 ( 1-\varepsilon) ( - \Delta  + \mathcal{N} ) - C_\varepsilon \leq  H_\alpha  \leq ( 1+\varepsilon) ( - \Delta  + \mathcal{N} ) +C_\varepsilon .
\end{align}
for all $\alpha \geq \alpha_0$. 
\end{lemma}

The proof is given in \cite{frankschlein} (see Lemma 7 and p.7) and relies on arguments of Lieb and Yamazaki \cite{liebyamazaki}; see \cite[p.12]{liebthomas} for a concise explanation.


\section*{Acknowledgments}

N.\,L.\  and R.\,S.\ gratefully acknowledge financial support by the European Research Council (ERC) under the European Union's Horizon 2020 research and innovation programme (grant agreement No 694227). B.\,S.\ acknowledges support from the Swiss National Science Foundation (grant 200020\_172623) and from the NCCR SwissMAP. N.\,L.\ would like to 
thank Andreas Deuchert and David Mitrouskas for interesting discussions. B.\,S. \ and R.\,S.\ 
would like to thank Rupert Frank for stimulating discussions about the time-evolution of a polaron.


{}

\vspace{0.5cm}

\noindent
(Nikolai Leopold) Institute of Science and Technology Austria (IST Austria)\\ Am Campus 1, 3400 Klosterneuburg, Austria\\ 
current address: University of Basel, Department of Mathematics and Computer Science \\
Spiegelgasse 1, 4051 Basel, Switzerland \\
E-mail address: \texttt{nikolai.leopold@unibas.ch} \\

\noindent
(Simone Rademacher) Institute of Mathematics, University of Zurich \\
Winterthurerstrasse 190, 8057 Zurich, Switzerland \\
current address: Institute of Science and Technology Austria (IST Austria)\\ Am Campus 1, 3400 Klosterneuburg, Austria\\ 
E-mail address: \texttt{simone.rademacher@ist.ac.at} \\

\noindent
(Benjamin Schlein) Institute of Mathematics, University of Zurich \\
Winterthurerstrasse 190, 8057 Zurich, Switzerland \\
E-mail address: \texttt{benjamin.schlein@math.uzh.ch} \\

\noindent
(Robert Seiringer) Institute of Science and Technology Austria (IST Austria)\\ Am Campus 1, 3400 Klosterneuburg, Austria\\ 
E-mail address: \texttt{robert.seiringer@ist.ac.at}

\end{document}